\documentclass[11pt,a4paper]{article}
\usepackage{amssymb,amsmath, amsfonts}

\usepackage{stix}
\usepackage{graphicx,graphics}
\usepackage{mathtools}
\usepackage[english]{babel}
\usepackage[utf8]{inputenc}
\usepackage{csquotes}
\usepackage{epsfig,url} 
\usepackage{bbm}
\usepackage{theorem} 
\usepackage{a4wide}
\usepackage{enumerate}
\usepackage{verbatim} 
\usepackage{color}
\usepackage{esint}

\usepackage[T1]{fontenc}

\newcommand{\ud}{\;\mathrm{d}} 
\newcommand{\uud}{\mathrm{d}}

\providecommand{\K}{\mathcal{K}}
\providecommand{\Pe}{\mathcal{P}}

\providecommand{\ol}{\overline}

\providecommand{\eps}{\varepsilon}
\providecommand{\vphi}{\varphi}

\providecommand{\supp}{\operatorname{supp}}
\providecommand{\exp}{\operatorname{exp}}
\providecommand{\La}{\mathcal{L}}
\providecommand{\F}{\mathcal{F}}

\newtheorem{theorem}{Theorem}[section]
\newtheorem{corollary}[theorem]{Corollary}
\newtheorem{proposition}[theorem]{Proposition}
\newtheorem{definition}[theorem]{Definition}
\newtheorem{lemma}[theorem]{Lemma}
\newtheorem{notation}[theorem]{Notation}

{\theorembodyfont{\upshape}
\newtheorem{remark}[theorem]{Remark}

}
\numberwithin{equation}{section}
\numberwithin{theorem}{section}
\newcommand{\qed}{\hfill$\Box$}
\newenvironment{proof}{\begin{trivlist}\item[]{\em Proof:}\/}{\qed\end{trivlist}}
\newenvironment{proofof}[1][Proof]{\noindent \textit{#1.} }{\ \qed}

\newcommand{\E}{{\mathbb E}}
\newcommand{\Z}{{\mathbb Z}}
\newcommand{\Reals}{{\mathbb R}}
\newcommand{\Complex}{{\mathbb C\hspace{0.05 ex}}}
\newcommand{\Naturals}{{\mathbb N}}

\newcommand{\cf}{{\mathbbm 1}}

\newcommand{\D}{\mathcal{D}}

\title{Convergence to the Landau equation from the truncated BBGKY hierarchy in the weak-coupling limit}

\author{ 
 Raphael Winter \thanks{\emailalessia}  \\[1em]  
\UBaddress}
\date{\today}
\newcommand{\email}[1]{E-mail: \tt #1}

\newcommand{\emailalessia}{\email{raphaelwinter@iam.uni-bonn.de}}
\newcommand{\UBaddress}{\em University of Bonn, Institute for Applied Mathematics\\
\em Endenicher Allee 60, D-53115 Bonn, Germany}

\date{\today}
	
\begin{document} 
  
\maketitle 
\begin{abstract} 
	We consider the evolution of the one-particle function in the weak-coupling limit in three space dimensions, obtained by truncating the BBGKY hierarchy under a propagation of chaos approximation. For this dynamics, we rigorously show the convergence to a solution
	of the Landau equation, keeping the full singularity of the Landau kernel. This resolves the issue arising from \cite{velazquez_non-markovian_2018}, where the singular region has been removed artificially. Since the singularity appears in the Landau equation due to the geometry of particle interactions, it is an intrinsic physical property of the weak-coupling limit which is crucial to the understanding of the transition from particle systems to the Landau equation.
\end{abstract}

\tableofcontents

\newpage


\section{Introduction} \label{Sec:Introduction}

The derivation of the Landau equation in the so-called weak-coupling limit is an open problem in kinetic theory. In the following we prove that the derivation is valid for short times, if we truncate the BBGKY hierarchy under the assumption of propagation of chaos. This requires the stability of the leading order nonlinear evolution up to macroscopic times of order one. A similar result has been proved for microscopic times in \cite{bobylev_particle_2013}, showing consistency with the desired limit. Furthermore, in \cite{velazquez_non-markovian_2018} the result is proved for short macroscopic times, when the interaction of particles with small relative velocity is removed. The objective of this paper is to prove the result without this cutoff.

To fix ideas, we first recall the weak-coupling limit. Consider a countable collection $(X_j,V_j)_{j\in J}$ of particles in the three dimensional phase space $\Reals^3\times \Reals^3$, randomly distributed according to a Poisson point process with intensity measure $\lambda(\uud{x}\uud{v})=\mu_\eps u_0(v) \ud{x}\uud{v}$. Here $u_0(v)$ is a probability density and $\eps>0$ is the scaling parameter. In particular, the random distribution of particles is homogeneous in space and there are on average $\mu_\eps$ many particles in a spatial unit volume. We fix a rotationally symmetric potential $\phi(x)$, and define the rescaled potential $\phi_\eps(x)=\phi(x/\eps)$. Now consider the following Newtonian dynamics:
\begin{align*}
	\dot{X}_j(t) = V_j(t), \quad \dot{V}_j(t) = -\eps^\frac12 \sum_{k\in J} \nabla \phi_\eps(X_j(t)-X_k(t)). 
\end{align*}
If we rescale the density of particles as $\mu_{\eps}= \mu \eps^{-3}$, the trajectories of particles
are governed by a large number of small deflections. It is expected that a central limit theorem for the collisions holds, so that in the limit we observe diffusion in the velocity variable. More precisely, let $u_\eps(t,x,v)$ be the density of particles defined by the expected number $n(t,A)$ of particles in a set $A\subset\Reals^3 \times \Reals^3$ at time $t\geq0$: 
\begin{align}\label{eq:oneparticle} 
	\int_{A} u_\eps(t,x,v)\ud{x}\ud{v}= \E[n(t,A)].
\end{align}  
Due to the homogeneity of the intensity measure $\lambda$ in space, we have $u_\eps(t,x,v)=u_\eps(t,v)$.
It is believed that in the scaling limit $\eps\rightarrow 0$ as introduced above, we have $u_\eps \rightarrow u$, where $u$ is a solution to the spatially homogeneous Landau equation:
\begin{equation} \label{eq:Landau}
\begin{aligned}
\partial_t u(t,v) 	&= \sum_{i,j=1}^3 \partial_{v_i}  \left(\int_{\Reals^3} a_{i,j}(v-v')
(\partial_{v_j} - \partial_{v'_j})\big( u(t,v) u(t,v') \big)\ud{v'}  \right) \\
u(0,v)	&= u_0(v).
\end{aligned}
\end{equation} 
For physical and mathematical justifications of this equation we refer to \cite{bobylev_particle_2013,landau_kinetische_1936,lifshitz_course_1981,spohn_kinetic_1980}.
The matrix $a_{ij}$ can be explicitly expressed by the interaction potential $\phi$:
\begin{align} \label{eq:adef} 
a_{i,j}(w) &= \frac{\pi^2}{4}\int_{\Reals^3} k_ik_j \delta(k\cdot w) |\hat{\phi}(k)|^2 \ud{k} 	
= \frac{\Lambda}{|w|} \left(\delta_{i,j}-\frac{w_i w_j}{|w|^2}\right) \quad \text{for some $\Lambda>0$}.
\end{align} 
Contrary to the case of the Boltzmann equation, a rigorous proof of the convergence $u_\eps\rightarrow u$ to the Landau equation has not been obtained so far.  
The heuristic justification of the weak-coupling limit to the Landau equation is based on 
the propagation of chaos principle. The principle, which is also crucial in the theory of the Boltzmann equation, asserts that with a high probability particles are uncorrelated prior to collision, and thereby experience a sequence of independent random deflections. 
More precisely, one can generalize the function $u_\eps$ defined in \eqref{eq:oneparticle} to functions $u_{\eps,n}$ describing the distribution of $n$-tuples of particles. Then the principle can be stated as:
\begin{align} \label{eq:propchaos}
	u_{\eps,n}(t,x_1,v_1,\dots,x_n,v_n) \approx \prod_{i=1}^n u_\eps (t,x_i,v_i).
\end{align}
Of course, \eqref{eq:propchaos} only holds for $\eps\rightarrow 0$, since particles will develop correlations through interaction. 
Controlling the propagation of chaos is a crucial step to obtain
a full derivation of the Landau equation, but we will not attempt to prove this rigorously here. 

If we truncate the system on the level of correlations of $n+1$-tuples of particles, we obtain
an approximation of the one-particle density in the weak-coupling limit that
is expected to be accurate to the $n$-th order in $\eps\rightarrow 0$,  for details see \cite{velazquez_non-markovian_2018}. Hence the leading order dynamics 
in $\eps\rightarrow 0$ can be obtained by only keeping pair correlations. This dynamics is
given by the equation:
\begin{equation} \label{eq:Main}
\begin{aligned} 
\partial_t u_\eps 	&= 	\frac1\eps \nabla_v \cdot \left(\int_0^t K[u_\eps(s)]\big(\frac{t-s}\eps,v\big) \nabla u_\eps(s,v)  
-  P[u_\eps(s)]\big(\frac{t-s}\eps,v\big) u_\eps(s,v)  \ud{s}\right)  \\ 
u_\eps(0,v)	&= u_0(v) , \\
K[u](t,v)			&:=	 \int \int \nabla \phi(x) \otimes \nabla \phi(x-t(v-v'))	u(v')  \ud{v'}\ud{x}\\
P[u](t,v)			&:=	\nabla_v \cdot K(t,v) = \int \int \nabla \phi(x) \otimes  \nabla\phi(x-t(v-v'))
\nabla u(v')  	\ud{v'}	\ud{x}.
\end{aligned}
\end{equation}
Our goal here is to show that the solutions of \eqref{eq:Main} converge to a solution
of the Landau equation \eqref{eq:Landau} as $\eps\rightarrow 0$.
Up to microscopic times, the result can be found in \cite{bobylev_particle_2013}, that is for $t\leq \eps$ as $\eps\rightarrow 0$.
In \cite{velazquez_non-markovian_2018}  the convergence is shown for times of order one, when $K$ is modified by putting a cutoff for $v-v'\approx 0$.

A crucial feature of the weak-coupling limit is the singularity $|w|^{-1}$ of the interaction
kernel $a_{i,j}$, see \eqref{eq:adef}. We stress the fact that the singularity appears in the limit, independent of the choice of interaction potential $\phi$. This can be explained considering the  momentum transfer of two interacting particles with velocities $v,v'$.
The duration of this transfer is proportional to the inverse relative velocity $|v-v'|^{-1}$, hence diverges for particles with very small relative velocity. This singularity is a key technical problem in the theory of the Landau equation and of the weak-coupling limit, see \cite{bobylev_particle_2013,guo_landau_2002,silvestre_upper_2017}. 
Furthermore, similar singularities appear in grazing collision limits from the Boltzmann equation to the Landau equation, and a number of equations with varying exponents of the singularity
have been studied in the literature (see: \cite{alexandre_landau_2004,desvillettes_spatially_2000,desvillettes_spatially_2000-1,villani_new_1998,villani_spatially_1998}). 
 
In this paper, we  prove the limit from \eqref{eq:Main} to the
Landau equation \eqref{eq:Landau}, keeping this important physical property of the system.
The technique presented here shows that singular operators of the form appearing
in the non-Markovian system \eqref{eq:Main} can be controlled using only the average-in-time dissipation of the equation that was proved and used in \cite{velazquez_non-markovian_2018}.

As in \cite{velazquez_non-markovian_2018}, we assume the system is initially close to the Maxwellian $m(v)$ and take the explicit potential $\phi$ with Fourier transform:
\begin{align} \label{FTpotential}
	\hat{\phi}(k) &= \frac1{(1+|k|^2)^\frac{3}{2}}.
\end{align}
The main result of this paper reads as follows. For the precise definition
of the function spaces, see Subsection~\ref{sec:func}.
\begin{theorem} \label{thm1}
	Let $m_0,\sigma>0$ and $m(\sigma^2,m_0)$ be the Maxwellian distribution with mass $m_0$ and standard deviation $\sigma$:
	\begin{align} \label{defmaxwellian}
	m(\sigma^2,m_0)(v):= m_0 \frac{e^{-\frac12 \frac{|v|^2}{\sigma^2}}}{(\sigma \sqrt{2\pi})^3}.
	\end{align}
	Let $n\geq 12$ and $v_0 \in H^n_{\lambda}$ satisfy: 
	\begin{align*}
	0 \leq v_0(v) \leq C e^{-\frac12 |v|}.  
	\end{align*}
	There exist $\delta_1,\eps_0 \in (0,\frac12]$  such that for all  $\eps,\delta_2 \in (0,\eps_0]>0$  the equation
	\begin{equation} \label{thmepseq}
	\begin{aligned} 
	\partial_t u_\eps 	= 	&\frac1\eps \nabla \cdot \left(\int_0^t K[u_\eps(s)]\left(\frac{t-s}\eps,v\right) \nabla u_\eps(s,v) \ud{s}\right) \\
	-  	&\frac1\eps \nabla \cdot \left(\int_0^t P[u_\eps(s)]\left(\frac{t-s}\eps,v\right) u_\eps(s,v)  \ud{s}\right)  \\ 
	u_\eps(0,\cdot)	&= u_0(\cdot) = m(v) + \delta_2 v_0(v)
	\end{aligned}	
	\end{equation}
	has a strong solution $u_\eps \in C^1([0,\delta_1]; H^{n-2}_{\lambda})$ up to time $\delta_1$.
	Furthermore, along a sequence $\eps_j\rightarrow 0$ the $u_{\eps_j}$ converge weakly to 
	$u_{\eps_j} \rightharpoonup^* u$ in $\D'$, and 
	the function $u \in  C^1([0,\delta_1];H^{n-4}_{\lambda})$ solves the Landau equation up to times $0\leq t\leq \delta_1$:
	\begin{equation} \label{uequation}
	\begin{aligned}
	\partial_t 	u 		&= \nabla \cdot \left( \K[u] \nabla u\right) - \nabla \cdot \left( \Pe[u]  u\right) \\
	u(0,v)			&= m(v) + \delta_2 v_0(v)  	\\
	\K[u](v)				&=  \frac{\pi^2}{4}\int (k \otimes k) |\hat{\phi}(k)|^2 \delta(k\cdot(v-v')) u(v') \ud{k} \ud{v'} \\
	\Pe[u](v)			&=  \frac{\pi^2}{4}\int (k \otimes k) |\hat{\phi}(k)|^2 \delta(k\cdot(v-v'))  \nabla u(v') \ud{k} \ud{v'}. 			
	\end{aligned}
	\end{equation}
\end{theorem}
An analogous result can be found in \cite{velazquez_non-markovian_2018}, however the interactions of particles with small relative
velocity is cut. More precisely, the result is shown for modified kernels $K,P,\K, \Pe$ with a cutoff function $\eta(v-v')$ in the integrals in $v'$, where $\eta$ is a smooth function that cuts off at the origin.

We now discuss the difficulty associated to removing this cutoff. Heuristically, our technique for proving an a priori estimate for solutions of \eqref{thmepseq}, independent of $\eps>0$, works as follows.  We multiply the equation
\begin{align*}
	\partial_t  u_\eps = \frac1\eps  \nabla \cdot \left(\int_0^t K[u_\eps(s)]\left(\frac{t-s}\eps,v\right) \nabla u_\eps(s,v) - P[u_\eps(s)]\left(\frac{t-s}\eps,v\right) u_\eps(s,v)  \ud{s}\right) 
\end{align*} 
with $u_\eps(t,v)$ and integrate in time and space. This yields an estimate for $u_\eps\in L^2(\Reals^+ \times \Reals^3)$, provided we can estimate two terms of the form:
\begin{align}
	Q_1 &= \frac1\eps\int_0^\infty \int_{\Reals^3} \nabla  u_\eps(t,v)   \left(\int_0^t K[u_\eps(s)]\left(\frac{t-s}\eps,v\right) \nabla u_\eps(s,v)\right)\ud{v} \ud{t} \label{Q1} \\
	Q_2 &=-\frac1\eps\int_0^\infty \int_{\Reals^3}  \nabla u_\eps(t,v) \left(\int_0^t \nabla \cdot K[u_\eps(s)]\left(\frac{t-s}\eps,v\right) u_\eps(s,v)\right) \ud{v} \ud{t}. \label{Q2}
\end{align}
The key point in \cite{velazquez_non-markovian_2018} is to prove an estimate of the form $Q_1\geq D_1>0$, where $D_1$ is the square of some (relatively weak) weighted $L^2(\Reals\times \Reals^3)$ norm for the Laplace-transform  of $\nabla u_\eps$. We stress the fact that such an estimate does not hold
for general kernels $K$, but is a feature very specific to the kernel emerging
from the weak-coupling limit. Since we prove $Q_1$ to have the good sign, this part of the argument is not affected by removing the cutoff.

It turns out that we can extract an a priori estimate for $u_\eps$, if we can show that for $c\in(0,1)$ there exists $C>0$ such that $|Q_2|\leq c D_1 +C\|u\|_{L^2(\Reals^+\times \Reals^3)} $. Such an estimate is difficult to obtain
without the cutoff in the space of relative velocities. The problem can be illustrated using  limit kernel $\K$ in \eqref{uequation} as an example, which yields terms of the form
\begin{align} \label{eq:pbterm}
	\nabla \cdot K[u_\eps(s,\cdot )](r/\eps,v) \sim \int \frac{B(r/\eps,v-v')u_\eps(s,v')}{|v-v'|^2} \ud{v'}.
\end{align}
Here $B(s,v-v')$ is some smooth vector-valued function.
Since we only expect $u_\eps(t,v) \in L^2(\Reals^+\times \Reals^3)$ the integral in
\eqref{eq:pbterm} cannot be expected to be in $L^\infty$. 
Instead we use that  $Q_2$ (cf. \eqref{Q2}) involves an integral in time. To obtain a bound for $Q_2$ we therefore carefully study the properties  of time integrals of
products of: 1) $\nabla u_\eps$ which has Laplace transform bounded in the weighted $L^2$ norm given by $D_1$, 2) $\nabla \cdot K[u_\eps(s)]$, which is a vector-valued
function of the form \eqref{eq:pbterm}, and 3) the function $u_\eps\in L^2(\Reals^+\times \Reals^3)$. This is the key step of this paper and the subject of Section~\ref{sec:Critical}.

\section{Structure of the proof}  

In this section we present the steps of the proof of Theorem \ref{thm1}.
The proofs of the individual steps are carried out in the main body of this paper.
The crucial novelty in our result is the a priori estimate for the system with the full singularity, which we prove in Section~\ref{sec:Critical}. Results that are only minor adaptations of the
result in \cite{velazquez_non-markovian_2018} can be found in the Appendix.
Only Lemmas whose proofs are not affected by the singularity of the kernel will be taken from \cite{velazquez_non-markovian_2018} without repeating the proof.

 We start by introducing the framework
of function sets and spaces that we will use throughout the paper.

\subsection{Functional setting and notation} \label{sec:func} 
	\begin{notation} \label{not:mollifyer}
	Let $\vphi_\gamma$, $\gamma>0$ be a standard mollifier sequence on $\Reals^3$. We set ${}^\gamma \nabla f(v) := \nabla (\vphi_\gamma * f)$ for $\gamma\in(0,1]$, and ${}^0 \nabla f=\nabla f $. With $K$ as in \eqref{eq:Main} we set: 
	\begin{align} \label{eq:Pgamma}
	P_\gamma = {}^\gamma \nabla \cdot K.
	\end{align}
	For the Laplace and the Fourier transform we use the conventions:
	\begin{align}
	\La(u)(z) 	&= \int_0^\infty u(t) e^{-z t} \ud{t}, \\
	\hat{u}(k)		&= \frac{1}{(2\pi)^{3/2}} \int_{\Reals^3} u(v) e^{-i k \cdot v} \ud{v}.
	\end{align}
\end{notation} 
We recall Plancherel's identity for Laplace transforms:
\begin{lemma} \label{Plancherel}
	Let $A\geq1$ and  $\mu_A(\mathrm{d}{t}) := e^{-At} \ud{t}$. Then for $u,v \in L^2(\mu_A)$ we have:
	\begin{align} \label{eq:Plancherel}
		(2\pi)^{\frac12} \int_0^\infty e^{-At} \ol{u}(t) v(t) \mu_A(\mathrm{d}t) = \int_\Reals
		\ol{\La(u)}\left(\frac{A}{2}+i\omega\right)
		\La(v) \left(\frac{A}{2}+i\omega\right) \ud{\omega}.
	\end{align}
	Furthermore, we will write $\langle \cdot, \cdot \rangle$ for the complex scalar product given by:
	\begin{align}
		\langle X, Y \rangle = \sum_{i=1}^n \sum_{j=1}^m \ol{X}_{i,j} Y_{i,j}, \quad  \text{ for $X,Y \in \Complex^{n\times m}$}.
	\end{align}
\end{lemma}
For the one-particle function $u_\eps$ we will use the following spaces.
\begin{definition} \label{def:spaces}
	Let $\lambda(v), \tilde{\lambda}(v)$ be defined by $\lambda(v):=e^{|v|}$, $\tilde{\lambda}(v):=\frac{e^{|v|}}{1+|v|}$.
	For $n \in \Naturals$ and $\nu=\lambda,\tilde{\lambda}$, we define $H^n_\nu$ as the closure of $C^\infty_c\big(\Reals^3\big)$
	with respect to the norm:
	\begin{align}\label{Hnnorm}
	\|u\|^2_{H^{n}_{\nu}} 		&:= \sum_{\alpha \in \Naturals^3, |\alpha|\leq n} \|\nu^\frac12 (\cdot) \nabla^\alpha u(\cdot)\|^2_{L^2}. 
	\end{align}
	For time-dependent functions $f(t,v)$, define the space $V^{n}_{A,\nu}$ as
	the closure of $C^\infty_c \big([0,\infty) \times \Reals^3)$ with respect to:
	\begin{align} \label{Vndef}
	\|f\|^2_{V^{n}_{A,\nu}} 	&:= \int_0^\infty e^{-At} \sum_{j=1}^d \|f_j(t,\cdot)\|^2_{H^n_\nu}  \ud{t}, \quad \text{where $A\geq 1$.} 
	\end{align}	
\end{definition}

\begin{definition}[Domain of the fixed point mapping]
	Fix $n\in \Naturals$ with $n\geq 12$. We define the norms:
	\begin{align}
		\|f\|_{E_A}&= \sup_{v\in \Reals^3, z \in \Complex: \Re(z)=A/2, |\beta|\leq n-6} |D^\beta_v\La(f)(z,v)(1+|z|^2)|e^{\frac12 |v|},  \\
		\|f\|_{F_{\eps,A}} &=  \sup_{v\in \Reals^3,z\in \Complex: \Re(z)=\frac{A}{2}, |\beta|\leq n-6}|D^\beta_v \La(f)(z,v)\frac{(1+|z|^2)(1+\eps|z|)}{\eps |z|}|e^{\frac12|v|}, \\
		\|f\|_{G_{\eps,A}} &=  \sup_{v\in \Reals^3,z\in \Complex: \Re(z)=\frac{A}{2}, |\beta|\leq n-6}|D^\beta_v \La(f)(z,v)(1+|z|^2)(1+\eps|z|)|e^{\frac12|v|}, \\
		\|f\|_{X^n_{A}}&= \|f\|_{V^n_{A,\tilde{\lambda}}}+\|\partial_t f\|_{V^{n-2}_{A,\tilde{\lambda}}} .
	\end{align}
	Let $\Gamma^{n}_{A,\delta_1} \subset V^n_{A,\tilde{\lambda}} $ be the set given by
	\begin{align} \label{def:Gamma}
		\Gamma^{n}_{A,\delta_1} = \{f \in V^n_{A,\tilde{\lambda}}: \int_{\Reals^3} f(t,v)\ud{v}=0, \supp f \subset [0,2\delta_1]\times \Reals^3  \}.
	\end{align}
	Let $\Omega^{A,\delta}_{\tilde{R},\delta_1,R,\eps}\subset V^n_{A,\tilde{\lambda}}$ be the set of functions given by:
	 \begin{align}
	\begin{aligned} \label{omegadef}
	\Omega^{A,\delta}_{\tilde{R},\delta_1,R,\eps} = \{f_1+f_2&=f\in \Gamma^{n}_{A,\delta_1}, \,\|f\|_{X^n_{A}}\leq 1  , \\
	&\sup_{t\in[0,1]}    \|\partial_t  f(t,\cdot)\|_{H^{n-4}_{\tilde{\lambda}}}\leq \tilde{R},\,  \|f\|_{G_{\eps,A}} \leq R, \,\|f_2\|_{F_{\eps,A}}\leq R, \,   \|f_1\|_{E_A}\leq \delta\}.	
	\end{aligned}  
	\end{align}
\end{definition}
\begin{lemma}[Properties of the domain] \label{lem:properties}
	For all $n\in \Naturals$, $n\geq 12$ and all $A\geq 1$, $\delta,\tilde{R},\delta_1,R,\eps>0$, the domain
	$\Omega^{A,\delta}_{\tilde{R},\delta_1,R,\eps}\subset V^n_{A,\tilde{\lambda}}$ is closed, convex and nonempty.
\end{lemma}
\begin{proof}
	The set $\Omega^{A,\delta}_{\tilde{R},\delta_1,R,\eps}$ is an intersection of convex sets, hence convex. Furthermore, it contains the zero function, hence it is nonempty.
	To see that the set is also closed, we first remark that $\Gamma^{n}_{A,\delta_1}$ is closed.
	Now	take a sequence $f_n \in \Omega^{A,\delta}_{\tilde{R},\delta_1,R,\eps}$ which is convergent
	in $V^n_{A,\tilde{\lambda}}$, that is:	$f_n \rightarrow f$ in $V^n_{A,\tilde{\lambda}}$ for some $f\in V^n_{A,\tilde{\lambda}}$. 
	We now observe that
	\begin{align*}
		 \|f\|_{X^n_A} \leq R, \quad \sup_{t\in[0,1]}    \|\partial_t  f(t,\cdot)\|_{H^{n-4}_{\tilde{\lambda}}}\leq \tilde{R}, \quad \|f\|_{G_{\eps,A}}\leq R
	\end{align*}
	by weak-* sequential compactness of the spaces generated by these norms.
	Furthermore, every element of the sequence can be decomposed into $f_n=f_{n,1}+f_{n,2}$ with
	\begin{align*}
		\|f_{n,1}\|_{E_A}\leq \delta, \quad \|f_{n,2}\|_{F_{\eps,A}}\leq R. 
	\end{align*}
	Now weak-* sequential compactness implies that such a decomposition also exists for the limit $f$, so we have $f\in \Omega^{A,\delta}_{\tilde{R},\delta_1,R,\eps}$ as claimed.	
\end{proof}

	We proceed by introducing  weight functions, that will allow us later to
keep track of the fine regularity and decay properties emerging from the evolution \eqref{eq:Main}.
The definitions can also be found in \cite{velazquez_non-markovian_2018}, we include them here to keep the analysis self-contained.
\begin{notation} \label{not:weights}
	For $z\in \Complex$ and $v \in \Reals^3$ define: 
	\begin{align}
	\alpha(z,v) &:= \frac{|\Im(z)|}{1+|v|},\quad		\beta(z,v)	:= \frac{|\Re(z)|}{1+|v|}.
	\end{align}
	Further we define positive functions $C_1$ and $C_2$ by:
	\begin{align}
	C_1(z,v)	&= 	\frac{1}{(1+|v|)(1+\alpha(z,v))^2} \label{C1} \\
	C_2(z,v)	&= 	\frac{\beta(z,v)+\alpha(z,v)^2 }{(1+|v|)(1+ \alpha(z,v))^4} \label{C2}.
	\end{align}	
	Let  $0 \neq v \in \Reals^3$, $V,W\in \Complex^3$. We define the anisotropic norm
	\begin{align}
	|W|_v := |P_{v}^\perp W| + \frac{|P_v W|}{1+|v|}  \label{Vnrom}, 
	\end{align}
	and weight functions $B_1(z,v)(V,W)$, $B_2(z,v)(V,W)$ by
	\begin{align}
	B_1(V,W) = C_1(z,v)|V|_v |W|_v +  C_2(z,v) |P_v V| |P_v W|, \label{quadrC} 
	\\		B_2(V,W) = C_1(z,v)|V|_v |W|_v +  C_3(z,v) |P_v V| |P_v W|. \label{quadrCu}	
	\end{align}
\end{notation}
Finally, in order to localize to short times, we introduce a family of cutoff functions.
\begin{definition} \label{def:cutoff}
	Let $\kappa\in C^\infty_c(\Reals;[0,1])$ be a cutoff function with
	$\kappa(s)=1$ for $s\in[-1,1]$ and	$\kappa(s)=0$ for $|s|\geq 2$.
	We set $\kappa_{\delta_1}$ to be the rescaled functions given by
	\begin{align}
	\kappa_{\delta_1}(s) := \kappa\left(\frac{s}{\delta_1}\right). \label{defkappa}
	\end{align}
\end{definition}

\subsection{Proof of the main result} \label{subsec:mainpf}
\begin{proofof}[Proof of Theorem~\ref{thm1}]
We start by restating the existence of a solution to  equation \eqref{thmepseq} as a fixed point problem. To this end, we introduce $f = u_\eps-u_0$, so \eqref{thmepseq} can be rewritten as:
	\begin{equation} \label{eq:fixedpoint}
		\begin{aligned} 
			\partial_t u_\eps 	= 	&\frac1\eps \nabla \cdot \left(\int_0^t K[u_0+f(s)]\left(\frac{t-s}\eps,v\right) \nabla u_\eps(s,v) \ud{s}\right) \\
			-  	&\frac1\eps \nabla \cdot \left(\int_0^t P[u_0+f(s)]\left(\frac{t-s}\eps,v\right) u_\eps(s,v)  \ud{s}\right)  \\ 
			u_\eps(0,\cdot)	&= u_0(\cdot).
		\end{aligned}	
	\end{equation}	 
	Finding a solution $u_\eps$ of \eqref{thmepseq} is equivalent to finding a fixed point of
	the mapping $f\mapsto u_\eps-u_0$, where $u_\eps$ is the solution of \eqref{eq:fixedpoint}
	for $f$ given.
	Since we cannot show the existence of \eqref{eq:fixedpoint} directly, we first consider a mollified version of the equation.
	With ${}^\gamma \nabla$, $P_\gamma$ as introduced in Notation~\ref{not:mollifyer}, consider the mollified equation:
		\begin{equation} \label{eq:fixedmollified}
		\begin{aligned} 
			\partial_t u_\eps 	= 	&\frac1\eps {}^\gamma \nabla \cdot \left(\int_0^t K[u_0+f(s)]\left(\frac{t-s}\eps,v\right) {}^\gamma \nabla u_\eps(s,v) \ud{s}\right) \\
			-  	&\frac1\eps {}^\gamma \nabla \cdot \left(\int_0^t  P_\gamma[u_0+f(s)]\left(\frac{t-s}\eps,v\right) u_\eps(s,v)  \ud{s}\right)  \\ 
			u_\eps(0,\cdot)	&= u_0(\cdot).
		\end{aligned}	
		\end{equation}	
	The existence of solutions to \eqref{eq:fixedmollified} can be proved in a straightforward fashion making use of the smoothness
	provided by the mollified gradient ${}^\gamma \nabla$.
	For the details of the argument we refer to \cite{velazquez_non-markovian_2018}. The result is stated in the following lemma, in the functional setting introduced in Subsection~\ref{sec:func}. Notice that this only establishes the existence of solutions for $\eps,\gamma>0$, but does not include an a priori estimate
	that is independent of these parameters.
	\begin{lemma}[Existence of a solution to the mollified equation] \label{lemexistence}
		Let $n\in \Naturals$, $\gamma,\eps>0$ and $u_0 \in H^n_{\lambda}$. Further
		assume $f\in L^1_{\operatorname{loc}}$, $\supp f \subset [0,1]$, and let $C>0$ such that $|f(t,v)| \leq C e^{-\frac12 |v|}$.  Then there exists a (unique)
		global solution $u_\eps \in C^1([0,\infty);H^n_{\lambda})$ to
		\eqref{eq:fixedmollified}.
	\end{lemma}
	Recall the definitions introduced in Subsection~\ref{sec:func}.
	We set up the mapping $\psi_{\delta_1}$ defined as:
	\begin{equation} \label{psidef}
		\begin{aligned}
		\Psi_{\delta_1}: 	\Omega^{A,\delta}_{\tilde{R},\delta_1,R,\eps} 	&\longrightarrow V^n_{A,\lambda} \\
		f						&\mapsto (u_\eps-u_0)k_{\delta_1}\text{, $u_\eps$ solution to \eqref{eq:fixedmollified}}.
		\end{aligned}
	\end{equation}
	The intuition to the various parameters appearing in \eqref{psidef} is the following:
	the paramater $A\geq 1$ determines the	exponential weight for large times and will be chosen large later. The  parameters $\tilde{R},R,\delta>0$ (cf. Definition~\ref{def:cutoff}) determine an a priori smallness assumption on the Laplace transform of $f$. Finally, $\delta_1>0$ is used to cut off to short times and can be used as an additional small parameter.

	The crucial point of the proof and the main content of this paper is the priori estimate for $\psi_{\delta_1}$, uniform in both the scaling parameter $\eps\rightarrow 0$ and the mollifying parameter $\gamma \rightarrow 0$. In \cite{velazquez_non-markovian_2018}, we proved the estimate after artificially removing the singular part of the appearing integral. We prove the following a priori estimate in this paper.
	\begin{theorem}[A priori estimate]  \label{thm:apriori}
		Let $n\geq 12$ and assume $u_0\in H^n_{\lambda}$ satisfies:
		\begin{align*}
		c \cf_{|v|\leq 4}(v) \leq u_0(v) \leq C e^{-\frac12 |v|}.
		\end{align*}	
		Then there exist $A,\delta,\delta_0>0$ s.t. for all $R,\tilde{R}>0$, $\delta_1\in(0,\delta_0)$ there is an $\eps_0>0$ such that
		for all $\gamma \in (0,\frac12]$ and 
		$\eps \in (0,\eps_0)$  the mapping $\psi_{\delta_1}$ introduced 
		in \eqref{psidef} is well-defined and continuous with respect to the topologies of $V^n_{A,\tilde{\lambda}}$, $V^n_{A,\lambda}$. Moreover, the mapping $\Psi_{\delta_1}$ satisfies:
		\begin{align} 
		\|\Psi_{\delta_1} (f)\|_{V^{n}_{A,\lambda}} &\leq 1, \quad 	\|\partial_t \Psi_{\delta_1} (f)\|_{V^{n-2}_{A,\lambda}} \leq 1. \label{PsiEst}
		\end{align}
	\end{theorem}

The theorem above ensures that \eqref{eq:fixedmollified} has
a solution for a given function  $f$ and provides an a priori estimate
that is uniform in $\eps,\gamma>0$. In the next step, we prove the existence
of a solution to
	\begin{equation} \label{eq:nonlinearfixed}
	\begin{aligned} 
	\partial_t u_\eps 	= 	&\frac1\eps {}^\gamma \nabla \cdot \left(\int_0^t K[u_\eps(s)]\left(\frac{t-s}\eps,v\right) {}^\gamma \nabla u_\eps(s,v) \ud{s}\right) \\
	-  	&\frac1\eps {}^\gamma \nabla \cdot \left(\int_0^t  P_\gamma[u_\eps(s)]\left(\frac{t-s}\eps,v\right) u_\eps(s,v)  \ud{s}\right)  \\ 
	u_\eps(0,\cdot)	&= u_0(\cdot),
	\end{aligned}	
	\end{equation}	
by applying Schauder's fixed point theorem to the mapping $\psi_{\delta_1}$.
To this end we need to show that, for an appropriate choice of the parameters,
the mapping $\psi_{\delta_1}$ defined in \eqref{psidef} leaves
the set $\Omega^{A,\delta}_{\tilde{R},\delta_1,R,\eps}$ invariant.
With the a priori estimate in Theorem~\ref{thm:apriori} at hand, the proof works similar to the corresponding proof in \cite{velazquez_non-markovian_2018}, and can therefore be found in the appendix. 
\begin{lemma}[Invariance of the domain] \label{thm:invariance} 
	Let  $v_0\in H^n_{\lambda}$, $n\geq 12$ be a function bounded above and below as $0 \leq v_0(v) \leq C e^{-\frac12 |v|}$, and  $u_0=m(v)+ \delta_2 v_0(v)$ for some $\delta_2>0$.
		
	Then there exist constants $A,\delta,\delta_1,\eps_0,R,\tilde{R}>0$ such that for all	$\delta_2,\eps \in (0,\eps_0]$,  $\gamma \in (0,1]$ the mapping $\Psi_{\delta_1}$ 
	defined in \eqref{psidef} leaves the set $\Omega^{A,\delta}_{\tilde{R},\delta_1,R,\eps}$ invariant.
\end{lemma}
The final ingredient to apply Schauder's fixed point theorem is the following compactness lemma. We do not perform the proof here, but refer to Lemma 2.5 in \cite{velazquez_non-markovian_2018}.
\begin{lemma}[Compactness] \label{lem:Rellich}
	Let $n\in \Naturals$ and let $(u_i)_{i\in \Naturals} \subset V^{n+1}_{A,\lambda}$ be a bounded sequence, such that the sequence
	$(\partial_t u_i) \subset V^{n+1}_{A,\lambda}$ is bounded as well. Then the sequence $(u_i)_{i\in \Naturals}$ is
	precompact in $V^n_{A,\tilde{\lambda}}$. 
\end{lemma}

With Lemma~\ref{lem:Rellich} and Lemma \ref{thm:invariance} at hand, the proof of the main result Theorem~\ref{thm1} quickly follows.
By Lemma~\ref{lem:properties}, the set $\Omega^{A,\delta}_{\tilde{R},\delta_1,R,\eps}$ is nonempty and convex. Further, by Lemma~\ref{thm:invariance} it is left invariant by the mapping $\psi_{\delta_1}$. Furthermore, for $\gamma>0$, the mapping is compact by Lemma~\ref{lem:Rellich}, so by Schauder's theorem there is a fixed point of $\psi_{\delta_1}$.
This yields the existence of solutions of \eqref{eq:nonlinearfixed}, which by Theorem~\ref{thm:apriori} are uniformly bounded in the parameters $\eps,\gamma$. Theorems~\ref{thm1} now simply follows from the
compactness shown in Lemma~\ref{lem:Rellich},  see also Section 5 in \cite{velazquez_non-markovian_2018}. 

Our proof of Theorem~\ref{thm1} is subject to the validity of Theorem~\ref{thm:apriori} and Lemma~\ref{thm:invariance}. The proof of Theorem~\ref{thm:apriori} is the content of the next section, the proof of Lemma~\ref{thm:invariance} can be found in the appendix.

\end{proofof}

\subsection{Proof of the a priori estimate}  
\begin{proofof}[Proof of Theorem~\ref{thm:apriori}]
In order to obtain uniform estimates for solutions of \eqref{eq:fixedmollified}, we differentiate the equation by $D^\alpha$. We then multiply with $e^{-At} D^\alpha u_\eps $, $A\geq 1$ and the weight function $\lambda(v)$ (cf. Definition~\ref{def:spaces}), and integrate in time to obtain: 
\begin{align} \label{eq:apriori}
A \int_0^T  \int \lambda(v) |D^{\alpha}u_\eps(t,v)|^2 e^{-At} \ud{t} \ud{v} \leq - 2 Q^{\alpha}_{\eps,A}[u_0+f](u_\eps \cf_{[0,T]}) + \|\lambda^\frac12 D^\alpha u_0\|^2_{L^2}. 
\end{align}
Here $Q^{\alpha}_{\eps,A}$ is an operator that decomposes into:
\begin{align}
Q_{\eps,A}^\alpha[\nu ](u) = &\sum_{\beta \leq \alpha} \binom{\alpha}{\beta} Q_{\eps,A}^{\alpha,\beta} [\nu ](u), \label{Qdef} \\
Q_{\eps,A}^{\alpha,\beta}[\nu](u) = 	& \int_0^\infty \int \frac{e^{-At}}{\eps}    {}^\gamma \nabla (D^\alpha u(t)\lambda)   \int_0^t  D^{\alpha-\beta} K[\nu(s)](\frac{t-s}{\eps})  {}^\gamma \nabla D^\beta u(s) \ud{s}\ud{v} 	\ud{t} \label{Kterm}\\
-	&\int_0^\infty\int \frac{e^{-At}}{\eps}    {}^\gamma \nabla (D^\alpha  u(t) \lambda)  \int_0^t  D^{\alpha-\beta} {}^\gamma \nabla \cdot K[\nu(s)](\frac{t-s}{\eps})		 D^\beta u(s) \ud{s}\ud{v} \ud{t}.	\label{Pterm}
\end{align}
We observe that the operators $Q_{\eps,A}^{\alpha,\beta}[\nu](u)$ are linear in the first argument and
quadratic in the second argument.
Using the linearity in the first argument, we rewrite $Q_{\eps,A}^\alpha[u_0+f](u)$ as:
\begin{align}
	Q_{\eps,A}^\alpha[u_0+f](u) &= Q_{\eps,A}^{\alpha,\alpha}[u_0 ](u) + Q_{\eps,A}^{\alpha,\alpha}[f ](u)	+ \sum_{\beta < \alpha} \binom{\alpha}{\beta} Q_{\eps,A}^{\alpha,\beta} [u_0 ](u) + \sum_{\beta < \alpha} \binom{\alpha}{\beta} Q_{\eps,A}^{\alpha,\beta} [f ](u).
\end{align}

Then we show that the first term has the correct sign, i.e. it yields a dissipative term after integrating in time, and the other terms can be handled as a perturbation.
More precisely, the a priori estimate in Theorem~\ref{thm:apriori} is a corollary of the following results.

\begin{lemma} \label{coercivitylemma}
	Let $n\geq 12$ and $u_0 \in H^n_\lambda$ be a function that is bounded above and below by
	\begin{align} \label{eq:u0est}
	c \cf_{|v|\leq 4}(v) \leq u_0(v) \leq C e^{-\frac12|v|}, \quad \|u_0\|_{H^n_{\lambda}} \leq C, \quad \text{for some constants $c,C>0$}.	
	\end{align}
	Let $A\geq 1$, $\eps \in  (0,\frac1{A}]$, $\gamma \in (0,1]$ and $\alpha \in \Naturals^3$
	be a multi-index of absolute value at most $n\in \Naturals$.
	Set $a= \frac{A}{2}$ and let $D^{\alpha}_{\eps,A}$ be given by ($z=a+i\omega$):
	\begin{align} \label{dissipexplic}
	D^{\alpha}_{\eps,A}(u):=  \int_\Reals \int_{\Reals^3}  B_1(\eps z,v) [{}^\gamma \nabla D^\alpha \La (u)(z,v),{}^\gamma \nabla D^\alpha \La (u)(z,v)]  \lambda(v) \ud{v} \ud{\omega} .
	\end{align}	
	Then for some constants $c,C>0$ independent of $A,\eps,\gamma$ we have:
	\begin{align}
	Q^{\alpha,\alpha}_{\eps,A}[u_0](u) \geq c D^{\alpha}_{\eps,A}(u) - C \|u\|_{V^n_{A,\lambda}}^2 					\label{quadrcoerc}. 
	\end{align}
\end{lemma}
The proof of Lemma~\ref{coercivitylemma} can be executed exactly as the proof of Lemma 3.7 in
\cite{velazquez_non-markovian_2018}. We therefore omit the proof here, and only shortly discuss the idea below. 
\begin{lemma} \label{lem:nonlinhigh}
	Let $A\geq 1$, $n\geq 12$, $\alpha$ a multi-index with $|\alpha|\leq n$ and $c>0$ arbitrary be given. 
	There exists $\delta_0(c,A,n)>0$ such that for all $\delta \in (0,\delta_0]$ and $R,\tilde{R}>0$, $\delta_1 \in (0,1)$
	we can estimate:
	\begin{align}
	(2\pi)^\frac12|Q^{\alpha,\alpha}_{\eps,A}[f](u)|	&\leq  c D^{\alpha}_{\eps,A}(u) +  \|u\|_{V^n_{A,\lambda}}^2, \label{freethmhigh}
	\end{align}
	for all $f \in \Omega^{A,\delta}_{\tilde{R},\delta_1,R,\eps}$, when $0<\eps\leq \eps_0(\delta,R,\tilde{R},\delta_1,A,c,n)$ is small.	
\end{lemma} 
The proof of Lemma~\ref{lem:nonlinhigh} requires only minor modifications from the one in \cite{velazquez_non-markovian_2018} and can therefore be found in the appendix.

\begin{lemma} \label{lem:lowerorder}
	Let $n\geq 12$  and  $c>0$ an arbitrarily small constant. There exists $C>0$ such that for all $A\geq 1$, $\alpha\in \Naturals^3$, $|\alpha|\leq n$ there
	exists $\delta_0(c,A,n)>0$ such that for all $\delta,\delta_1 \in (0,\delta_0]$ and all $\tilde{R},R>0$, 
	we can estimate:
	\begin{align}
		\sum_{\beta<\alpha} \binom{\alpha}{\beta}|Q^{\alpha,\beta}_{\eps,A}[u_0](u)|&\leq c D^\alpha_{\eps,A}(u)+ C \|u\|^2_{V^n_{A,\lambda}}, \label{quadrbili}\\ 
	\sum_{\beta < \alpha} \binom{\alpha}{\beta} |Q^{\alpha,\beta}_{\eps,A}[f ](u)|	&\leq  c D^{\alpha}_{\eps,A}(u) +  C\|u\|_{V^n_{A,\lambda}}^2, \label{freethmlow}
	\end{align}
	for all $f \in \Omega^{A,\delta}_{\tilde{R},\delta_1,R,\eps}$, when $0<\eps\leq \eps_0(\delta,\delta_1,R,\tilde{R},A,c_1,n)$ is small.	
\end{lemma} 
When keeping the full singularity of the kernel $K$, the critical terms are the $Q_{\eps,A}^{\alpha,\beta}$ with $\beta\approx \alpha$, since then the derivatives
act on the singular kernel $K$. Therefore, the proof of Lemma~\ref{lem:lowerorder} requires new ideas and is a key point of this paper. The proof is carried out in Section~\ref{sec:Critical}.

 Since the dissipativity proved in Lemma~\ref{coercivitylemma} is crucial to the understanding of the a priori estimate, we briefly sketch the idea of the proof here.
We rewrite $Q^{\alpha,\alpha}_{\eps,A}$  as:
\begin{align*}
Q_{\eps,A}^{\alpha,\alpha}[u_0](u) = 	& \int_0^\infty \int \frac{e^{-At}}{\eps}    {}^\gamma \nabla D^\alpha u(t)\lambda   \int_0^t   K[u_0](\frac{t-s}{\eps})  {}^\gamma \nabla D^\alpha u(s) \ud{s}\ud{v} 	\ud{t}\\
+	& \int_0^\infty \int \frac{e^{-At}}{\eps}   D^\alpha u(t) {}^\gamma \nabla (\lambda)   \int_0^t   K[u_0](\frac{t-s}{\eps})  {}^\gamma \nabla D^\alpha u(s) \ud{s}\ud{v} 	\ud{t}\\
-	&\int_0^\infty\int \frac{e^{-At}}{\eps}    {}^\gamma \nabla (D^\alpha  u(t) \lambda)  \int_0^t   {}^\gamma \nabla \cdot K[u_0](\frac{t-s}{\eps})		 D^\alpha u(s) \ud{s}\ud{v} \ud{t}.	
\end{align*}
The crucial point to handle $Q^{\alpha,\alpha}_{\eps,A}$ is the first term on the right-hand
side above. 
Using Lemma~\ref{Plancherel}, the time integration transforms into ($V=V(z,v)= \La({}^\gamma \nabla D^\alpha u)(z,v)$):
\begin{align*}
& (2\pi)^{\frac12}\int_0^\infty \int_{\Reals^3} \frac{e^{-At}}{\eps}    {}^\gamma \nabla D^\alpha u(t)\lambda   \int_0^t   K[u_0](\frac{t-s}{\eps})  {}^\gamma \nabla D^\alpha u(s) \ud{s}\ud{v} 	\ud{t}	
\\
= &\int_\Reals \int_{\Reals^3}  \langle V\lambda  ,  \La( K)[u_0](\eps z) V \rangle \ud{v} 	\ud{\omega}
\end{align*} 
where $z=a+i\omega=A/2+i\omega$. Since the function $\La(K)$ is pointwise a symmetric matrix, we can symmetrize the expression above and obtain:
\begin{align*}
& (2\pi)^{\frac12}\int_0^\infty \int \frac{e^{-At}}{\eps}    {}^\gamma \nabla D^\alpha u(t)\lambda   \int_0^t   K[u_0](\frac{t-s}{\eps})  {}^\gamma \nabla D^\alpha u(s) 
=\int_\Reals \int   \langle V\lambda ,   \Re(\La( K))[u_0](\eps z) V \rangle \ud{v} 	\ud{\omega}.
\end{align*}
The particular kernel $\La(K)$ given by the identity \eqref{KLapl} has the property
that the real part $\Re(\La(K))$ is complex coercive if $u_0$ is a positive function.
A careful analysis of this coercivity property yields a lower
bound for $Q^{\alpha,\alpha}_{\eps,A}[u_0](u)$ in terms
of the anisotropic weight functional $B$ introduced in \eqref{quadrC} and yields Lemma~\ref{coercivitylemma}.



Due to the explicit choice of the potential $\phi$ (cf. \eqref{FTpotential}), some
integrals appearing in this analysis become explicit. More precisely,
we make use of the subsequent lemma, for the computation we refer to \cite{velazquez_non-markovian_2018}.
\begin{lemma} \label{Mlemma}
	Define matrix-valued functions $M_1(z,v),M_2(z,v)$ by
	\begin{equation}\label{defM}
	\begin{aligned}
	M_1(z,v)	&:= \frac{\pi^2}{4 |v|} \frac{1}{1+\frac{z}{|v|}} P_v^\perp, \quad 			 			M_2(z,v)	:= \frac{\pi^2}{4 |v|} \frac{ \frac{z}{|v|}}{(1+\frac{z}{|v|})^2} P_{v},\quad  \text{ for $\Re(z)\geq 0$, $v\in \Reals^3$} .		
	\end{aligned}
	\end{equation} 
	Further let $u \in H^n_{\tilde{\lambda}}$ (cf. \eqref{Hnnorm}), then the following identity holds:
	\begin{align} \label{Midentity}
	\int (k \otimes k) |\hat{\phi}(k)|^2 \frac{z}{z^2+(k\cdot v)^2} \ud{k} &=  M_1(z,v)+M_2(z,v)\\
	\La(K[u])(z,v) 	&=  \int (M_1+M_2)(z,v-v') u(v')  \ud{v'} \label{KLapl} .
	\end{align}
\end{lemma}
The following Proposition is the analog of Lemma 3.3 in \cite{velazquez_non-markovian_2018}, here however with the singularity in $v$ kept.
\begin{proposition}
	Let $\beta\in \Naturals^3$ be a multi-index and $M_i$, $i=1,2$ as introduced in \eqref{defM}. Then there exists a constant $C_{|\beta|}>0$, such that for $V,W\in \Complex^3$ arbitrary complex vectors, $z\in \Complex$ with $\Re(z)\geq0$ and $v\in \Reals^3$ we have:
	\begin{align} \label{eq:Mestimate}
	|\langle V, D^\beta(M_i(z,v)) W\rangle | \leq  \frac{C_{|\beta|}|V||W|}{|v|^{|\beta|+1} (1+\alpha(z,v))}.
	\end{align}
\end{proposition}
\begin{proof}
	Using that $\frac{P_v}{|v|}, \frac{P^\perp_v}{|v|}$ are functions with homogeneity $-1$, the
	estimate follows by explicit differentiation of the remaining function.
\end{proof}
\begin{lemma} \label{anisotrope}
	Let $z\in \Complex$ with $0\leq \Re(z)\leq 1$ and $\beta$ be a multi-index.
	Let $V,W \in \Complex^3$ be  complex vectors. Further let $n\geq 1$ and $f \in H^n_{\tilde{\lambda}}$ satisfy $|f|\leq C e^{-\frac12 |v|}$.
	Recall $B_1$, $B_2$ as defined in \eqref{quadrC}-\eqref{quadrCu} and $C_1$ defined in \eqref{C1}.
	Then for $|\beta|\leq 1$, there holds
	\begin{align}
	\left|\int_{\Reals^3} \langle V,(M_1+M_2)(z,v-v')  W \rangle | f(v')  \ud{v'} \right|
	\leq C	&(1+\alpha(z,v)) B_2(z,v)(V,W)
	\label{matrixest2} \\
	\left|\int_{\Reals^3} \langle V, D^\beta \left((M_1+M_2)(z,v-v')\right) W \rangle  f(v')  \ud{v'} \right|	\leq C	&\frac{(1+\alpha(z,v))}{(1+|v|)^{|\beta|}} C_1(z,v)|V||W|.
	\label{matrixest3}
	\end{align}
\end{lemma}

With the Lemmas above at hand, we can conclude the proof of Theorem~\ref{thm:apriori}:
There exists a $C>0$ such that for all $A\geq 1$ and there exist $\delta,\delta_0>0$ s.t. for all $R,\tilde{R}>0$, $\delta_1\in(0,\delta_0)$ there is an $\eps_0>0$ such that
\begin{align*}
	|Q^\alpha_{\eps,A}(u_0 + f )(u)| \leq C \|u\|^2_{V^n_{A,\lambda}}. 
\end{align*}
Then the identity \eqref{eq:apriori} implies (possibly changing $C$)
\begin{align}
	A \|u_\eps\|^2_{V^n_{A,\lambda}} \leq C\|u_\eps\|^2_{V^n_{A,\lambda}} + \|u_0\|_{H^n_\lambda}^2. \label{eq:aprio1}
\end{align}
On the other hand, we can use the equation \eqref{eq:fixedmollified} to find an estimate for the time derivative. For any multi-index $\alpha$ with $|\alpha|\leq n$ we have:
\begin{equation} \label{eq:derivlaplace}
\begin{aligned}	
	\La (\partial_t D^\alpha u_\eps ) = &\sum_{\beta\leq \alpha}	\binom{\alpha}{\beta} {}^\gamma \nabla \cdot \left( \int (M_1+M_2)(\eps z,v-v') D^\beta u_0(v')  {}^\gamma \nabla \La(D^{\alpha-\beta}u_\eps \kappa_{2\delta_1})(z,v)\ud{v'} \right)  \\
	-&\binom{\alpha}{\beta} {}^\gamma \nabla \cdot \left( \int  {}^\gamma \nabla(M_1+M_2)(\eps z,v-v') D^\beta u_0(v')   \La(D^{\alpha-\beta} u_\eps \kappa_{2\delta_1})(z,v) \ud{v'} \right) \\
	+& \binom{\alpha}{\beta}{}^\gamma \nabla \cdot \left( \int (M_1+M_2)(\eps z,v-v')   {}^\gamma \nabla \La(D^\beta f(\cdot,v') D^{\alpha-\beta} u_\eps\kappa_{2\delta_1} (\cdot,v))(z) \ud{v'}\right)  \\
	-&  \binom{\alpha}{\beta}{}^\gamma \nabla \cdot \left( \int  {}^\gamma \nabla(M_1+M_2)(\eps z,v-v')\La( D^\beta f(\cdot,v') D^{\alpha-\beta} u_\eps(\cdot,v))(z) \ud{v'}  \right).
\end{aligned}
\end{equation}
Since $\|f\|_{X^n_{A}}\leq 1$, we can estimate $f$ uniformly in time as:
\begin{align*}
	\|f(t,\cdot)\|_{H^m_{\tilde{\lambda}}} \leq C, \quad \text{for $m\leq n-2$}.
\end{align*}
Bringing this to the equation \eqref{eq:derivlaplace} and using $|M_1(z,v)|+|M_2(z,v)|\leq C/|v|$ gives
\begin{align*}
	\|\partial_t  u_\eps\|_{V^l_{A,\lambda}} \leq C \| u_\eps\|_{V^n_{A,\lambda}}, \quad \text{for $l\leq n-4$}.
\end{align*} 
For $0< \delta_1\leq 1/A $ and $l\leq n-4$ this yields:
\begin{align*}
	\|u_\eps \kappa_{2\delta_1}(t,\cdot)\|_{H^l_\lambda} \leq C \| u_\eps\|_{V^n_{A,\lambda}} 
\end{align*}
We assume $n\geq 12$, so we know $|\beta|\leq n-4$ or $|\alpha-\beta|\leq n-4$. This allows to use \eqref{eq:derivlaplace} and infer
\begin{align*}
	\|\partial_t (u_\eps-u_0)\|_{V^{n-2}_{A,\lambda}}\leq C \|u_\eps\|_{V^{n}_{A,\lambda}}. 
\end{align*}
Furthermore, Poincaré inequality yields 
\begin{align}
\|\partial_t ((u_\eps-u_0)\kappa_{\delta_1})\|_{V^{n-2}_{A,\lambda}}\leq C \|\partial_t (u_\eps-u_0)\|_{V^{n-2}_{A,\lambda}},
\end{align}
with a constant independent of $\delta_1>0$.
 With \eqref{eq:aprio1} we conclude that we can pick $A\geq 1$ large enough such that \eqref{PsiEst} holds.

\end{proofof}
\section{A priori estimate for the non-Markovian system with the full kernel} \label{sec:Critical}

\subsection{Toolbox}

We start by providing some lemmas that will enable us to use equation \eqref{eq:fixedmollified} in Laplace variables to bootstrap estimates on the solution $u_\eps$. To this end, we introduce the following notation.
\begin{notation}
	The convolution $*_a$ is to be understood as ($z=a+i\omega$): 
	\begin{align}\label{eq:convolution}
	(f*_a g)(z)=(f*_a g)(a+i\omega)= \frac{1}{2\pi}\int_{\Reals} f(i\theta) g(a+i(\omega-\theta))\ud{\theta}.
	\end{align}
\end{notation}

\begin{lemma} \label{lem:lowerfreez}
	Let $A\geq 1$, $\delta_1\in (0,1/A)$, $f \in \Gamma^{n}_{A,\delta_1}$ (cf. \eqref{def:Gamma}) with $\|f\|_{V^n_{A,\tilde{\lambda}}}\leq 1$, and assume $u$ solves equation \eqref{eq:fixedmollified}. Then for $m\leq n-4$ we have:
	\begin{align}
	 \|\|\La(D^m (u_\eps-u_0)\kappa_{2\delta_1})(z ,v)\|_{L^2_\lambda}\|_{L^1_{\Re(z)=0}}&\leq C  \|u\|_{V^n_{A,\lambda}}. \label{eq:uRe0}
	\end{align} 
\end{lemma} 
\begin{proof}
	By Plancherel's identity (cf. \eqref{eq:Plancherel}) and the assumption $0<\delta_1<1/A$ we have:
	\begin{align}
	\|\|\La(u_\eps\kappa_{2\delta_1})\|_{H^n_{\lambda}}\|_{L^2(i\Reals)}+\|\|\La(u_\eps\kappa_{2\delta_1})\|_{H^n_{\lambda}}\|_{L^\infty(i\Reals)} &\leq C \|u_\eps\|_{V^n_{A,\lambda}}.  \label{lowertool1}
	\end{align}
	Since $\supp f \subset [0,2 \delta_1]\times \Reals^3$ and $\|f\|_{V^n_{A,\tilde{\lambda}}}\leq 1$,  we also have:
	\begin{align}	
	\|\|\La(f)\|_{H^n_{\tilde{\lambda}}}\|_{L^2(i\Reals)}+\|\|\La(f)\|_{H^n_{\tilde \lambda}}\|_{L^\infty(i\Reals)} &\leq C  \label{lowertool2}.
	\end{align}
	We transform  equation \eqref{eq:fixedmollified} to Laplace variables (recall the notation $*_a$ introduced in \eqref{eq:convolution}):
	\begin{align*}
	z\La(u_\eps-u_0) = 	&\nabla \cdot \left(  \int (M_1+M_2)(\eps z,v')  \left(\La(u_0+f)(\cdot,v-v')*_a\nabla \La (u_\eps \kappa_{2 \delta_1})  (\cdot,v)\right)(z) \ud{v'} 	 \right) \\
	-&\nabla \cdot \left(  \int (M_1+M_2)(\eps z,v')  \left(\nabla\La(u_0+f)(\cdot,v-v')*_a \La (u_\eps \kappa_{2 \delta_1})(\cdot,v)\right)(z)  \ud{v'}	 \right).
	\end{align*}
	Here we have taken advantage of the fact that $f$ is compactly supported to localize $u_\eps$.
	We use that the matrices $M_1,M_2$ (cf. \eqref{defM}) satisfy $|(M_1+M_2)(z,v)| \leq 1/|v|$. 
	Hence we can estimate $u_\eps$ as follows:
	\begin{align} \label{bootroot}
	|z|\|\La((u_\eps-u_0) )(z)\|_{H^{k-2}_{\lambda}} \leq C \left( \|\La(u_\eps\kappa_{2\delta_1 })\|_{H^k_\lambda}+  \|\La(f)(\cdot)\|_{H^{k}_\lambda} *_{\Re(z)}\|\La(u_\eps\kappa_{2 \delta_1})(\cdot)\|_{H^k_\lambda}(z)\right).
	\end{align}
	Inserting the estimates \eqref{lowertool1} and \eqref{lowertool2} above yields:
	\begin{align} \label{eq:bootstrapstart}
		\|\|\La((u_\eps-u_0) )(z)\|_{H^{n-2}_{\lambda}}\|_{L^{\frac32}_{\Re(z)=0}} \leq C  \|u_\eps\|_{V^n_{A,\lambda}}.
	\end{align}
	We observe that the cutoff functions $\kappa_{2 \delta_1}$ satisfy:
	\begin{align} \label{eq:Fourcutoff}
		\|\F(\kappa_{2 \delta_1})(\cdot)\|_{L^1(\Reals)}\leq C	.
	\end{align}
	Hence we can infer from \eqref{eq:bootstrapstart} that $u_\eps$ also satisfies: 
	\begin{align} \label{eq:32 est}
	\|\|\La((u_\eps-u_0)\kappa_{2 \delta_1} )(z)\|_{H^{n-2}_{\lambda}}\|_{L^{\frac32}_{\Re(z)=0}} \leq C  \|u_\eps\|_{V^n_{A,\lambda}}.
	\end{align}
	Plugging \eqref{eq:32 est} back into \eqref{bootroot}, and using \eqref{eq:Fourcutoff} we find 
	\begin{equation*}
		\|\|\La(D^m (u_\eps-u_0)\kappa_{2\delta_1})(z ,v)\|_{L^2_\lambda}\|_{L^1_{\Re(z)=0}}\leq C  \|u\|_{V^n_{A,\lambda}}, \quad \text{ for $m\leq n-4$},
	\end{equation*}
	which concludes the proof of the Lemma.
\end{proof}

We have shown the proof of Theorem \ref{thm1} in Section~\ref{subsec:mainpf},
provided the validity of Lemma~\ref{lem:lowerorder}. Proving this result is the crucial part of the paper and the content of this section.

Henceforth we will use that $\frac{1}{|v|}$ is, up to a constant, the fundamental solution
to the Laplace equation on the whole space. This implies that 
the convolution operator appearing in the Landau equation \eqref{eq:Landau}, as well
as the operators appearing in the non-Markovian equation \eqref{eq:Main} can be studied in terms of the inverse of the Laplacian. In the following, we collect the corresponding estimates in the weighted space $H^n_\nu$ introduced in Definition~\ref{def:spaces}.

From the theory of the Laplace
equation, we have the following straightforward result.
\begin{lemma} \label{lem:greensop}
	For $u \in C^\infty_c(\Reals^3)$, let $T[u]$ be given by the convolution operator
	\begin{align} \label{eq:Tdef}
		T[u](v) = \int_{\Reals^3} \frac{u(v')}{|v-v'|} \ud{v'}.
	\end{align}
	Then $T$ can be extended to a continuous operator $T: L^2(\Reals^3) \longrightarrow \dot{H}^2(\Reals^3)$, i.e. for some $C>0$ we have:
	\begin{align*}
		\|T[u]\|_{\dot{H}^2} \leq C \|u\|_{L^2}.
	\end{align*}
\end{lemma}
\begin{remark}
	We write $\dot{H}^k(\Reals^3)=\dot{H}^k$ for the homogeneous Sobolev space of $k$-th order, i.e. the closure
	of $C^\infty_c(\Reals^3)$ with respect to the norm:
	\begin{align*}
		\|u\|_{\dot{H}^k} := \|\nabla^k u \|_{L^2(\Reals^3)}.
	\end{align*}
\end{remark}
%
\begin{lemma} \label{lem:l6sov}
	We have a Sobolev embedding for $\dot{H}^1$:
	\begin{align}\label{eq:sobolev}
		\|u\|_{L^6} \leq C \|u\|_{H^1}.
	\end{align}
	Furthermore, there is a constant $C>0$ such that for all $u\in H^1_{\lambda}$ we have 
	\begin{align}\label{eq:sobweighted}
		\|u(\cdot ) e^{\frac12|\cdot|}\|_{L^6} \leq C \|u\|_{H^1_\lambda}.
	\end{align}
\end{lemma}
\begin{proof}
	The estimate \eqref{eq:sobolev} is the classical Gagliardo-Nierenberg inequality. For the proof of \eqref{eq:sobweighted} 
	we write $u\in H^1_\lambda$ as $u= \tilde{u}(v)e^{-\frac12 |v|}$. Since 
	the weight function satisfies $|\nabla e^{-\frac12|v|}| \sim e^{-\frac12 |v|}$, we have
	\begin{align*}
		\|\tilde{u}\|_{H^1} \leq C \|u\|_{H^1_\lambda}.
	\end{align*}
	Now the claim follows from \eqref{eq:sobolev}.
\end{proof}
The main tool to handle the full singularity of the Landau kernel is contained in the following Lemma. 
Here we combine the classical results in the Lemmas above with the weighted spaces
given by the weight functions defined in Notation~\ref{not:weights}, to obtain $L^6_{\operatorname{loc}}$ estimates
that are compatible with the structure of the dissipation functional $D^{\alpha}_{\eps,A}(u)$ in
\eqref{dissipexplic}.

\begin{lemma} \label{lem:singularintegral}
	Recall the weight function $\alpha$ introduced
	in Notation~\ref{not:weights}. For $z\in \Complex$, let $e(z,\cdot) \in W^{1,1}(\Reals^3)$ be a function satisfying the estimate:
	\begin{align} \label{eq:eAss}
	|e(z,v)| &\leq \frac{C}{|v|(1+\alpha(z,v))}, \quad |\nabla e(z,v)| \leq \frac{C}{|v|^2(1+\alpha(z,v))} .
	\end{align}
	For $f\in L^2_{\tilde{\lambda}}$ define the function:
	\begin{align}
	E(f)(z,v) := \int_{\Reals^3} f(v') e(z,v-v') \ud{v'}.
	\end{align}
	Then the following holds
	\begin{align} \label{est:L6}
	\|(1+|v|)^2 (1+\alpha(z,v))\nabla_v E(f)\|_{L^p(B_1(v^*))} \leq C_p \|f\|_{L^2_{\tilde{\lambda}}}, \quad \text{for $1\leq p\leq 6$, $v^* \in \Reals^3$}.
	\end{align}
\end{lemma}
\begin{proof}
	We decompose the velocity space $\Reals^3\setminus \{0\}$ into annuli 
	\begin{align} \label{def:Aj}
		A_j:= \{v \in \Reals^3: 2^j \leq |v| <2^{j+1}\},\quad  j \in \Z.	
	\end{align}
	For $v\in \Reals^3$ a given vector, let $j(v)$ be such that $v \in A_{j(v)}$ and:
	\begin{align} \label{def:Iv}
			I(v)=\{j(v)-1,j(v),j(v)+1\} \subset \Z.	
	\end{align}   
	A function $f\in L^2_{\tilde{\lambda}}$ we
	write as $f(v)= \tilde{f}(v) e^{-\frac12 |v|}$. Notice that $\tilde{f}$ satisfies:
	\begin{align} \label{eq:tilde}
		\left\|\frac{\tilde{f}(v)}{(1+|v|)^\frac12}\right\|_{L^2} = \|f\|_{L^2_{\tilde{\lambda}}}.
	\end{align}
	Then we can  estimate the derivative of $E$ using \eqref{eq:eAss}:
	\begin{align*}
	|\nabla E(f)(z,v)|	&\leq \sum_{k \in \Z}\int_{A_j} \frac{|f(v')|}{|v-v'|^2 (1+\alpha(z,v-v')) } \ud{v'} \\
	\leq &\sum_{k \in \Z \setminus I(v)}\int_{A_j} \frac{|\tilde{f}(v')|e^{-\frac12 |v'|}}{|v-v'|^2 (1+\alpha(z,v-v') )} \ud{v'}
	+  \sum_{k \in I(v)}\int_{A_j} \frac{|\tilde{f}(v')|e^{-\frac12 |v'|}}{|v-v'|^2 (1+\alpha(z,v-v')) } \ud{v'} \\
	=:& Z_1 + Z_2 .
	\end{align*}  
	We estimate the terms $Z_1$ and $Z_2$ separately. Due to the dyadic scaling we have:
	\begin{align*}
		\frac{1}{|v-v'|^2 (1+\alpha(z,v-v'))} 			&\leq C (1+|v'|) \frac{1}{|v|^2(1+\alpha(z,v))} \quad &\text{on $A_k$ with $k>j(v)+1$}, \\
			\frac{1}{|v-v'|^2 (1+\alpha(z,v-v'))} 		&\leq C  \frac{1}{|v|^2(1+\alpha(z,v))}\quad &\text{on $A_k$ with $k<j(v)-1$}.
	\end{align*}
	Hence we can estimate:
	\begin{align} \label{logelei}
	\frac{1}{|v-v'|^2(1+\alpha(z,v-v'))} \leq C\frac{1+|v'|}{|v|^2(1+\alpha(z,v))}, \quad  \text{for $v'\in A_k, k\notin I(v)$}.
	\end{align}
	We use \eqref{logelei} to obtain an upper bound for $Z_1$:
	\begin{align*}
	|Z_1(z,v)|  &\leq C \sum_{j \in \Z \setminus I(v)}\int_{A_j} \frac{|\tilde{f}(v')|\exp(- 2^{j-3})}{(1+|v|)^2 (1+\alpha(z,v-v') )} \ud{v'} \\
				&\leq C \sum_{j \in \Z \setminus I(v)}\frac{\exp(- 2^{j-3})}{(1+|v|)^2 (1+\alpha(z,v) )}\int_{A_j} (1+|v'|) |\tilde{f}(v')| \ud{v'}.
	\end{align*}
	For $j\in \Z \setminus I(v)$ we estimate with Young's inequality:
	\begin{align*}
		\int_{A_j} (1+|v'|) |\tilde{f}(v')| \ud{v'} &\leq \left\|\frac{\tilde{f}(v)}{(1+|v|)^\frac12}\right\|_{L^2} \|  (1+|v'|)^\frac32)\|_{L^2(A_j)} \\
			&\leq (1+2^j)^2 2^{\frac32 j} \left\|\frac{\tilde{f}}{(1+|v|)}\right\|_{L^2}.
	\end{align*}
	We use the identity \eqref{eq:tilde} for $\tilde{f}$ to find
	\begin{equation}
	\begin{aligned}  \label{Z1est}
	 	|Z_1(z,v)|  	&\leq C \sum_{j \in \Z \setminus I(v)}\frac{\exp(- 2^{j-3})}{(1+|v|)^2 (1+\alpha(z,v) )} ((1+2^j)2^j)^\frac32  \left\|\frac{\tilde{f}}{(1+|v|)^\frac12}\right\|_{L^2}  \\
	 	&\leq  \frac{C\|f\|_{L^2_{\tilde{\lambda}}}}{(1+|v|^2)(1+\alpha(z,v))}.
	\end{aligned}
	\end{equation}
	It remains to estimate the term $Z_2$. To this end, for $v' \in A_j$, $j\in I(v)$ we have $\frac14 |v|\leq |v'|\leq 4 |v|$. This yields that the weight function $\alpha$ satisfies, for $v,v'$ as above:
	\begin{align}\label{eq:alphabd}
		\frac{1}{1+\alpha(z,v-v')} \leq C \frac{1}{1+\alpha(z,v)}.
	\end{align}
	Furthermore, for $v' \in A_j$, $j\in I(v)$ there holds:
	\begin{align}\label{eq:bd}
		e^{-\frac12 |v'|} \leq e^{-\frac18 |v|},\quad 	0<c \leq \frac{1+|v'|}{1+|v|}\leq C. 
	\end{align}
	Combining the estimates \eqref{eq:alphabd} and \eqref{eq:bd}, we obtain a bound for $Z_2$:
	\begin{align} \label{eq:Z2first}
	|Z_2(z,v)|\leq \frac{Ce^{-\frac18 |v|}(1+|v|)^\frac12}{1+\alpha(z,v)}\sum_{j \in I(v)}\int_{A_j} \frac{|\tilde{f}(v')|/(1+|v'|)^\frac12}{|v-v'|^2 } \ud{v'}. 
	\end{align} 
	The convolution can be estimated using Lemma \ref{lem:greensop} and \eqref{eq:tilde}:
	\begin{align*}
		\left\|\sum_{j \in I(v)}\int_{A_j} \frac{|\tilde{f}(v')|/(1+|v'|)^\frac12}{|\cdot-v'|^2 } \ud{v'}\right\|_{\dot{H}^1} &\leq C \left\|\frac{\tilde{f}(v)}{(1+|v|)^\frac12}\right\|_{L^2} = C \|f\|_{L^2_{\tilde{\lambda}}}.
	\end{align*}
	Then the Sobolev embedding \eqref{eq:sobolev} yields:
	\begin{align} \label{eq:singAj}
		\|\int_{A_j} \frac{|\tilde{f}(v')|/(1+|v'|)^\frac12}{|\cdot-v'|^2 } \ud{v'}\|_{L^6} \leq C \|f\|_{L^2_{\tilde{\lambda}}}.  
	\end{align}
	Bringing this estimate to \eqref{eq:Z2first}, we obtain: 
	\begin{align} \label{Z2est}
	\|(1+|v|)^2 (1+\alpha(z,v)) Z_2(z,v)\|_{L^p(B_1(v^*))} \leq C_p \|f\|_{L^2_{\tilde{\lambda}}}, \quad \text{for $1\leq p\leq 6$, $v^* \in \Reals^3$}.
	\end{align}
	Combining \eqref{Z1est} and \eqref{Z2est} gives the claim of the Lemma.
\end{proof}

\begin{definition}
	Let  $\Lambda$ is given by $M_1,M_2$ (cf. \eqref{defM}) as:
	\begin{align}
	\Lambda[f](z,\tau,v)	&= \int_{\Reals^3} \int_{0}^\infty  \left(M_1+M_2\right)(z,v-v') e^{-i\tau s}f(s,v') \ud{s} \ud{v'} \label{Lambdadef} .
	\end{align}
\end{definition}
Using the Plancherel identity \eqref{eq:Plancherel} the functionals $Q^{\alpha,\beta}_{\eps,A}$ can be expressed as:
\begin{equation}\label{eq:Qfrep}
\begin{aligned}
	(2\pi)^\frac12 Q^{\alpha,\beta}_{\eps,A}[f](u)	
	=	&\int_{\Reals} \int_\Reals \int \langle\nabla  ( \lambda   \La( D^\alpha u) (z)), \Lambda[f](\eps z, 				\omega-\theta) \La(\nabla D^\beta u)(p)\rangle \ud{v}\ud{\theta} \ud{\omega} \\
	-	&\int_{\Reals} \int_\Reals \int \langle  \nabla  (\La( D^\alpha u) \lambda)(z), \nabla\cdot  \Lambda[f](\eps 			z, \omega-\theta) \La(D^\beta u)(p) \rangle \ud{v}\ud{\theta} 	\ud{\omega}.
\end{aligned}
\end{equation}
We now prove that \eqref{est:L6} yields an estimate for quadratic functionals
of a certain form. This will allow us to obtain a bound for the
functionals $Q^{\alpha,\beta}_{\eps,A}$ in terms of the dissipation $D^\alpha_{\eps,A}$ (cf. \eqref{dissipexplic}) and the average-in-time
norm $V^n_{A,\lambda}$. 

\begin{corollary} \label{lem:L6}
	There exists $C>0$ such that for all $v^*\in \Reals^3$, $z\in \Complex$ with $\Re(z)\geq0$ and $|\mu|=1$:
	\begin{align}
	\|(1+\alpha(z,v))(1+|v|^2) D^\mu \La(K)[g]\|_{L^6(B_1(v^*))} &\leq C \|g\|_{L^2_{\tilde{\lambda}}}. \label{lin1}
	\end{align}
\end{corollary}
\begin{proof}
		By the identity \eqref{KLapl} we have:
	\begin{align} 
	\La(K[g])(z,v) 	&=  \int (M_1+M_2)(z,v-v') g(v')  \ud{v'},
	\end{align}
	where $M_1,M_2$  are given by \eqref{defM}. We now use \eqref{eq:Mestimate} to infer that
	$e(z,v):= (M_1+M_2)(z,v-v')$ satisfies the assumptions of Lemma \ref{lem:singularintegral}, which proves \eqref{lin1}. 
\end{proof}

\begin{corollary} \label{lem:kernelL2}
	For all $A\geq 1$ and  $v^*\in \Reals^3$, $z,p\in \Complex$, $\Re(z)\geq 0$, $\Re(p)\geq 0$, $|\mu|=1$:
	\begin{align}
	\|(1+\alpha(z,v))(1+|v|^2) D^\mu \Lambda[g](z,p,v)\|_{L^6(B_1(v^*))} 		&\leq C \|\La g(p,\cdot)\|_{L^2_{\tilde{\lambda}}} \label{nonlin1} 
	\end{align}
	
\end{corollary}

\begin{lemma} \label{lem:realFT}
	Let $f\in \Omega^{A,\delta}_{\tilde{R},\delta_1,R,\eps}$ with $A\geq 1$, $R,\tilde{R}>0$, $\eps,\delta,\delta_1 \in (0,1)$. Further let $|\nu|\leq n-6$.  Then we can estimate:
	\begin{align}
	\left|\int e^{-is\tau} \nabla^\nu f(s,v) \ud{s}\right|\leq C(A) \frac{e^{-\frac12 |v|}}{(1+|\tau|^2)} \min\{\delta+\frac{R\eps |\tau|}{(1+\eps|\tau|)},\frac{R}{(1+\eps|\tau|)}\}, \quad \text{for $\tau \in \Reals$}.
	\end{align}
\end{lemma}
\begin{proof}
	By definition of $\Omega^{A,\delta}_{\tilde{R},\delta_1,R,\eps}$, the estimate above holds for $\tau\in \Complex$ with $\Im(\tau)=-a=-A/2$. Further $f\in \Omega^{A,\delta}_{\tilde{R},\delta_1,R,\eps}$ yields
	$\supp f \subset [0,2]\times \Reals^3$, so $f(t,v)= f(t,v) \kappa_{2}(t)$, hence
	the estimate follows by the convolution identity for Laplace transforms.
\end{proof}

\begin{lemma}[Estimate for $\nu$ small] \label{lem:uregular}
	Let $\varrho$ be a function satisfying $\varrho(v)\leq Ce^{|v|}$.
	For any $c>0$ there exists $C>0$ such that for  $A\geq 1$ and $\delta_1\in (0,1/A)$ we have:
	If $f \in \Omega^{A,\delta}_{\tilde{R},\delta_1,R,\eps}$ and $u_\eps$ satisfies \eqref{eq:fixedmollified},
	then we have:
	\begin{align}
	\left|\int_{\Reals} \int_\Reals \int \langle \nabla D^\alpha \La u_\eps (z) \varrho , D^{\mu} \Lambda[f](\eps z, \omega-\theta) \La( D^\nu (u_\eps-u_0))(p)\rangle  \ud{v}\ud{\theta} \ud{\omega}\right|  &\leq c D^\alpha_{\eps,A}(u_\eps) + C \|u_\eps\|^2_{V^n_{A,\lambda}}, \label{est:low2}
	\end{align}
	for multi-indices $\alpha, \nu, \mu$ with $|\nu|\leq n-5$, $|\alpha|\leq n$, $1\leq|\mu|\leq n$. Here $z=a+i\omega$, $p=a+i\theta$.
\end{lemma}
\begin{proof} 
	We start by estimating the integral in $\theta$.
	We introduce the function $H_\eps(z,v)$ given by: 
	\begin{align*}
	H_\eps(z,v)&=\int  D^{\mu} \Lambda[f](\eps z, p) \La( D^\nu (u_\eps-u_0))(i(\omega-\theta))  \ud{\theta} \\
				&=\int  D^{\mu} \Lambda[f](\eps z, p) \La( D^\nu (u_\eps-u_0)\kappa_{2 \delta_1})(i(\omega-\theta))  \ud{\theta}.
	\end{align*}
	Furthermore, Corollary~\ref{lem:kernelL2} gives an estimate for $\Lambda[f]$: 
	\begin{align}\label{est:Lambdalocal}
	\|(1+\alpha(z,v'))(1+|v'|^2) D^\mu \Lambda[f](\eps z,p,v')\|_{L^6(B_1(v))} 		&\leq C \|\La(f)(p,\cdot)\|_{H^n_{\tilde{\lambda}}}.
	\end{align}
	Now we pick a Vitali covering of the space $\Reals^3$. More precisely, we cover the space
	with the collection of balls given by $\left(B_\frac13(v)\right)_{v\in \Reals^3}$. By Vitali's covering Lemma there is a sequence of balls , $(B_1(v_k))_{k\in \Naturals}\subset \Reals^3$ such that
	\begin{align*}
	\Reals^3 &= \bigcup_{k\in \Naturals} B_1(v_k),\quad 	B_\frac13(v_j) \cap B_\frac13(v_k)=\emptyset, \quad  \text{ for $j\neq k$}.
	\end{align*}
	Using that $	B_\frac13(v_k)$ are disjoint, the balls $B_k :=B_1(v_k)$ satisfy
	\begin{align} \label{summinglp}
	\left(\sum_{k\in \Naturals} \|f\|^p_{L^p(B_k)}\right)^{1/p} \leq  \|f\|_{L^p(\Reals^3)}.
	\end{align}
	Now we apply Sobolev embedding on the balls $B_k$:
	\begin{align} \label{est:LuLocal}
	&\|\La(D^\nu (u_\eps-u_0)\kappa_{2 \delta_1})(z ,v)e^{\frac12|v|}\|_{L^3(B_k)}
	\leq C\|\La(D^\nu (u_\eps-u_0)\kappa_{2 \delta_1})(z ,v)e^{\frac12 |v|}\|_{H^1(B_k)}.
	\end{align}
	Now we start estimating $H_\eps(z,v)$. With the estimate \eqref{summinglp} we obtain (here $p=a+i\theta$, $a=A/2$):
	\begin{align*}
	&\|\|(1+\alpha(z,v))(1+|v|^2)H_\eps(z,v)e^{\frac12 |v|}\|_{L^2(\Reals^3)}\|_{L^2_{\Re(z)=a}} \\
	\leq  C &\| \int_\Reals \|(1+\alpha)(1+|v|^2) D^\mu \Lambda[f](\eps z,p) \La(D^\nu (u_\eps-u_0)\kappa_{2 \delta_1})(z-p)e^{\frac12 |v|} \|_{L^2(\Reals^3)} \ud{\theta} \|_{L^2_{\Re(z)=a}} \\
	\leq  C &\| \int_{\Reals}\left(  \sum_k \|(1+\alpha)(1+|v|^2) D^\mu \Lambda[f](\eps z,p) \La(D^\nu (u_\eps-u_0)\kappa_{2 \delta_1})(z-p)e^{\frac12 |v|} \|_{L^2(B_k)}^2 \right)^\frac12 \ud{\theta} \|_{L^2_{\Re(z)=a}}.
	\end{align*}
	On each ball $B_k$ we apply Young's inequality and obtain:
	\begin{align*}
	&\left\|(1+\alpha)(1+|v|^2) D^\mu \Lambda[f](\eps z,p) \La(D^\nu (u_\eps-u_0)\kappa_{2 \delta_1})(z-p)e^{\frac12 |v|} \right\|_{L^2(B_k)}^2 \\
	\leq &\ \|(1+\alpha)(1+|v|^2) D^\mu \Lambda[f](\eps z,p)\|_{L^6(B_k)}^2\|\La(D^\nu (u_\eps-u_0)\kappa_{2 \delta_1})(z-p)e^{\frac12 |v|}\|^2_{L^3(B_k)} .
	\end{align*}
	Using $|\nu|\leq n-5$ and the estimates \eqref{est:Lambdalocal} and \eqref{est:LuLocal} on each ball $B_k$ we get:
	\begin{equation}
	\begin{aligned} \label{eq:Hepscomp}
		&\|\|(1+\alpha(z,v))(1+|v|^2)H_\eps(z,v)e^{\frac12 |v|}\|_{L^2(\Reals^3)}\|_{L^2_{\Re(z)=a}}\\
	\leq  C &\| \int_{\Reals} \|\La(f)(p)\|_{H^n_{\tilde{\lambda}}}\left(\sum_k  \|\La(D^\nu (u_\eps-u_0)\kappa_{2 \delta_1})(i(\omega-\theta))e^{\frac12 |v|}\|^2_{L^3(B_k)} \right)^\frac12 \ud{\theta}  \|_{L^2_{\Re(z)=a}} \\
	\leq  C &\| \int_{\Reals} \|\La(f)(p)\|_{H^n_{\tilde{\lambda}}} \left(\sum_k  \|\La( (u_\eps-u_0)\kappa_{2 \delta_1})(i(\omega-\theta))e^{\frac12 |v|}\|^2_{H^{n-4}(B_k)} \right)^\frac12  \ud{\theta}  \|_{L^2_{\Re(z)=a}}.	
	\end{aligned}
	\end{equation}
	By assumption, $f\in \Omega^{A,\delta}_{\tilde{R},\delta_1,R,\eps}$, therefore:
	\begin{align} \label{eq:fHnL2}
		\|\| \La(f)(p)\|_{H^n_{\tilde{\lambda}}} \|_{L^2_{\Re(p)=a}} \leq 1.
	\end{align}
	We insert the estimates \eqref{eq:fHnL2} and \eqref{summinglp} into \eqref{eq:Hepscomp}:
	\begin{align*}
	\|\|(1+\alpha(z,v))(1+|v|^2)H_\eps(z,v)e^{\frac12 |v|}\|_{L^2(\Reals^3)}\|_{L^2_{\Re(z)=a}}
	\leq C  \|\|\La( (u_\eps-u_0)\kappa_{2 \delta_1})(z)\|_{H^{n-2}_{\lambda}}\|_{L^1_{\Re(z)=0}}.
	\end{align*}
	Applying \eqref{eq:uRe0}, we conclude
	\begin{equation}\label{est:Heps1}
	\begin{aligned}			
	\|\|(1+\alpha(z,v))(1+|v|^2)H_\eps(z,v)e^{\frac12 |v|}\|_{L^2(\Reals^3)}\|_{L^2_{\Re(z)=a}} \leq C  \|u\|_{V^n_{A,\lambda}}. 
	\end{aligned}
	\end{equation}
	We recall the definition of the dissipation function $D^\alpha_{\eps,A}$ (cf. \eqref{dissipexplic}) and observe that:
	\begin{align}\label{est:dissbound}
	\| \frac{ \nabla  D^\alpha \La(u)}{(1+\alpha(\eps z,v))(1+|v|^2)}\|_{V^0_{A,\lambda}} \leq C D^\alpha_{\eps,A}.
	\end{align}
	Finally, combining the estimates \eqref{est:Heps1} and \eqref{est:dissbound} yields:
	\begin{align*}
	&\left|\int_{\Reals} \int_\Reals \int  \langle \nabla  D^\alpha \La(u_\eps) (z) \varrho , D^{\mu} \Lambda[f](\eps z, \omega-\theta) \La( D^\nu (u_\eps-u_0))(p) \rangle  \ud{v}\ud{\theta} \ud{\omega}\right|	\\
	&\leq \| \frac{ \nabla  D^\alpha \La(u)}{(1+\alpha(\eps z,v))(1+|v|^2)}\|_{V^0_{A,\lambda}} \|u\|_{V^n_{A,\lambda}} \leq  c D^\alpha_{\eps,A} + C \|u\|_{V^n_{A,\lambda}},	
	\end{align*}
	as claimed.
\end{proof}
\begin{lemma}[Estimate for $\nu$ large] \label{lem:higherterms}
	Let $\varrho(v)$ be a function satisfying $|\varrho(v)|\leq C e^{-\frac12 |v|}$.
	For $c>0$ given, there exists $C>0$ such that for all $A\geq 1$ there is a $\delta_0(A)>0$ small, such that $\forall \delta \in (0,\delta_0)$, $R>0$ we can choose $\eps>0$ small enough such that
	\begin{align}
	\left|\int \langle \varrho(v)  \nabla D^\alpha \La(u) ,D^{\mu} \Lambda[f](\eps z, \omega-\theta)  D^\nu \La(u)(p,v)\rangle  \ud{v} \ud{\omega} \ud{\theta} \right|&\leq  c D^\alpha_{\eps,A}(u) + C \|u\|^2_{V^n_{A,\lambda}},\label{high}
	\end{align}
	for all $|\alpha|,|\nu|\leq n$, $1\leq |\mu|\leq n-6$, $\delta_1\in (0,1)$, $\tilde{R}>0$ and $f\in \Omega^{A,\delta}_{\tilde{R},\delta_1,R,\eps}$.
\end{lemma}
\begin{proof}
	Let $A\geq 1$.
	Since $|\mu|\geq 1$, the estimate \eqref{eq:Mestimate} and Lemma~\ref{lem:realFT} allow to choose $\delta_0(A)$ such that for $\delta\in (0,\delta_0)$, and $R>0$ we have
	\begin{align*}
	|D^\mu \Lambda[f](\eps z,\omega-\theta) |\leq \frac{1}{(1+|v|^2)(1+\alpha(\eps z, v))(1+|\omega-\theta|^2)},
	\end{align*}
	provided $\eps>0$ is small enough. This implies that  we can bound
	\begin{align*}
	&\left|\int \langle \varrho(v)  \nabla D^\alpha \La(u) ,D^{\mu} \Lambda[f](\eps z, \omega-\theta)  D^\nu \La(u)(z,v) \rangle \ud{v} \ud{\omega} \ud{\theta} \right| \\
	\leq 	&\int \frac{\varrho(v)  |\nabla D^\alpha \La(u)|}{(1+|v|^2)(1+\alpha(\eps z, v))(1+|\omega-\theta|^2)} | D^\nu \La(u)(p,v)| \ud{v} \ud{\omega} \ud{\theta} 			\leq  c D^\alpha_{\eps,A} + C \|u\|_{V^n_{A,\lambda}},
	\end{align*}
	as was claimed.
\end{proof}

\subsection{Proof of Lemma~\ref{lem:lowerorder}}

\begin{proofof}[Proof of Lemma~\ref{lem:nonlinhigh}]
	We split the left-hand side of \eqref{freethmlow} as:
	\begin{align*}
		\sum_{\beta < \alpha} \binom{\alpha}{\beta} |Q^{\alpha,\beta}_{\eps,A}[f ](u)|	&= 
		\sum_{\beta < \alpha,|\beta| \leq  n-4} \binom{\alpha}{\beta} |Q^{\alpha,\beta}_{\eps,A}[f ](u)|+ \sum_{\beta < \alpha, |\beta|>n-4} \binom{\alpha}{\beta} |Q^{\alpha,\beta}_{\eps,A}[f ](u)|\\
		&=: I_1+I_2.
	\end{align*}
	We start by estimating $I_1$. 	To this end, rewrite the quadratic functionals 
	$Q^{\alpha,\beta}_{\eps,A}[f](u)$ using Plancherel's theorem:
	\begin{align}
	Q^{\alpha,\beta}_{\eps,A}[f](u)	= 	&(2\pi)^{-\frac12} (J_1^{\alpha,\beta} + J^{\alpha,\beta}_2),	
	\end{align}
	where $J^{\alpha,\beta}_1$ is given by
	\begin{equation}  \label{LapNon}
	\begin{aligned}
	J^{\alpha,\beta}_1 =	&\int_{\Reals} \int_\Reals \int \langle \nabla (D^\alpha \La(u) \lambda)(z), D^{\alpha-\beta}\Lambda[f](\eps z, \omega-\theta) \La(\nabla D^\beta (u-u_0))(p) 	\rangle \ud{v}\ud{\theta} \ud{\omega} \\
	+	\int_{\Reals} &\int_\Reals \int \langle \nabla (D^\alpha \La(u) \lambda)(z), \nabla\cdot  D^{\alpha-\beta}\Lambda[f](\eps z, \omega-\theta) \La(\nabla D^\beta (u-u_0))(p) 	\rangle \ud{v}\ud{\theta} \ud{\omega} ,
	\end{aligned}
	\end{equation}
	and $J_2$ by
	\begin{equation}
	\begin{aligned}
	J^{\alpha,\beta}_2 =	&\int_{\Reals}  \int \langle \nabla (D^\alpha  \La(u) \lambda)(z),  D^{\alpha-\beta} \int (M_1+M_2)(\eps z,v-v') \La(f)(z,v') \ud{v'}D^\beta u_0 \rangle 					\ud{v} \ud{\omega} \\
	+	\int_{\Reals}&  \int \langle \nabla (D^\alpha  \La(u) \lambda)(z), \nabla \cdot D^{\alpha-\beta} \int (M_1+M_2)(\eps z,v-v') \La(f)(z,v') \ud{v'}D^\beta u_0 \rangle 					\ud{v} \ud{\omega}.
	\end{aligned} 	
	\end{equation}
	We apply Lemma~\ref{lem:uregular} to infer
	\begin{align*}
		J_1^{\alpha,\beta} &\leq c D^\alpha_{\eps,A}(u) + C \|u\|^2_{V^n_{A,\lambda}}, \quad \text{for $|\beta|\leq n-6$} .
	\end{align*} 
	The bound for $J_2^{\alpha,\beta}$ follows similarly. By assumption $u_0\in H^n_\lambda$, and
	Corollary~\ref{lem:kernelL2} gives:
	\begin{align*}
		J_2^{\alpha,\beta}\leq c D^\alpha_{\eps,A}(u), \quad \text{for $|\beta|\leq n-4$}.
	\end{align*}
	Therefore, we have
	\begin{align}\label{eq:I1}
		I_1 \leq c D^\alpha_{\eps,A}(u) + C \|u\|^2_{V^n_{A,\lambda}}. 
	\end{align}
	It remains to prove the same estimate for $I_2$.
	We write $Q^{\alpha,\beta}_{\eps,A}[f](u)$ as:
	\begin{align*}
		 Q^{\alpha,\beta}_{\eps,A}[f](u)=	&\int_{\Reals} \int_\Reals \int \langle \nabla (D^\alpha \La(u) \lambda)(z), D^{\alpha-\beta}\Lambda[f](\eps z, \omega-\theta) \La(\nabla D^\beta u)(p) 	\rangle \ud{v}\ud{\theta} \ud{\omega} \\
		+	\int_{\Reals} &\int_\Reals \int \langle \nabla (D^\alpha \La(u) \lambda)(z), \nabla \cdot D^{\alpha-\beta}\Lambda[f](\eps z, \omega-\theta) \La(D^\beta u)(p) 	\rangle \ud{v}\ud{\theta} \ud{\omega}.
	\end{align*}
	Since  $|\beta|\geq n-5$, we have $|\alpha-\beta|\leq 5$. By assumption $n\geq 12$, hence we can apply Lemma~\ref{lem:higherterms} to get:
	\begin{align} \label{eq:I2}
		I_2 &\leq c D^\alpha_{\eps,A}(u) + C \|u\|^2_{V^n_{A,\lambda}}.
	\end{align}
	Combining \eqref{eq:I1} and \eqref{eq:I2} proves \eqref{freethmlow}.

\end{proofof}

\section{Appendix}
\subsection{Proof of Lemma~\ref{lem:nonlinhigh} }

\begin{proofof}[Proof of Lemma~\ref{lem:nonlinhigh}]
	To prove a bound for the functional $Q^{\alpha,\alpha}_{\eps,A}[f]$ in terms of the 
	dissipation $D^\alpha_{\eps,A}$ (cf. \eqref{dissipexplic}), we make
	use of its symmetry properties. The proof and the lemmas contained therein is similar to the proof
	of Lemma 4.9 in \cite{velazquez_non-markovian_2018}.
	
	We start by introducing some notation.
	For $\eps>0$, $v\in \Reals^3$, $z=a+i\omega, p=a+i\theta \in \Complex$, define the matrices $L_1$, $L_2$:
	\begin{align}
	L_1(\eps,z,p,v) &:=  	\frac12 ( M_1(\eps z, v) + M_1(\eps \ol{p}, v)) , \quad  \label{L12def}L_2(\eps,z,p,v) := 		\frac12 ( M_2(\eps z, v) + M_2(\eps \ol{p}, v)), 
	\end{align}
	and the  symmetrized kernel $\Lambda_s$ by (again writing $z=a+i\omega, p=a+i\theta$):
	\begin{equation}\label{Lambdas}
	\begin{aligned} 
	\Lambda_s[f ](\eps,z,p,v)&:= \Lambda_1[\nu ](\eps,z,p,v) + \Lambda_2[f ](\eps,z,p,v)    \\
	\Lambda_1[f ](\eps,z,p,v)&:=  \int_{\Reals^3 \setminus B_1(0)} L_1(\eps,z,p,v') \left(\int_0^\infty e^{-is(\omega-\theta)}f(s,v-v') \ud{s}\right) \ud{v'} \\
	\Lambda_2[f ](\eps,z,p,v)&:=  \int_{\Reals^3 \setminus B_1(0)} L_2(\eps,z,p,v') \left(\int_0^\infty e^{-is(\omega-\theta)}f(s,v-v') \ud{s}\right)  \ud{v'} .	
	\end{aligned}
	\end{equation}
	We split the kernel $L_2$ into two terms:
	\begin{align}
	N_2(\eps,z,p,v) 			&= L_2 (\eps,z,p,v)-N_1(\eps,z,p,v), \text{where} \label{N2}  \\
	N_1(\eps,z,p,v)								&= \frac{1}{|v|^2}\frac{\eps (a+i(\theta- \omega))}{(1+\frac{\eps z}{|			v|})^2(1+\frac{\eps \ol{p}}{|v|})^2} P_{v} 					\label{N1}.	
	\end{align} 	
	Further define $\Lambda_0(z,v)$ by  
	\begin{align}
		\Lambda_0[f](z,\tau,v)	&= \int_{B_1(0)} \int_{0}^\infty  \left(M_1+M_2\right)(z,v') e^{-i\tau s}f(s,v-v') \ud{s} \ud{v'} \label{Lambda0def} .
	\end{align}
We will make use of the following straightforward estimates (compare Lemma~4.6 in \cite{velazquez_non-markovian_2018}).
\begin{lemma} \label{matrixsymm}
	Let $a>0$ and $z=a+i\omega$, $p= a+ i\theta$, and  $0<\eps \leq \frac{1}{a}$. For $V,W \in \Complex^3$ and $L_1$ as introduced in \eqref{L12def}, we have
	\begin{equation}\label{est:L1bound}
	\begin{aligned}	
	|\langle V,	L_1(\eps,z,p,v)W \rangle|	&\leq C 	\frac{|P_{v}^\perp V| |P_{v}^\perp W|}{1+|v|} \frac{1+ \eps |	\theta-\omega|}{(1+\alpha(\eps z,v))(1+ \alpha(\eps p,v))}, \quad &\text{for $|v|\geq 1$}. 						 
	\end{aligned}
	\end{equation}
	Similarly we have an estimate for $N_2$:
	\begin{equation} \label{est:N2bound}
	\begin{aligned}
			| \langle V, N_2(\eps,z,p,v)  W \rangle |  	&\leq C 	\frac{|V||W|}{1+|v|^3} \frac{\eps^2 |p| |z| +\eps^2|p| |z| 				(1+\eps|\theta-\omega|)}						{(1+\alpha(\eps z,v))^2(1+ \alpha(\eps p,v))^2}  , \quad &\text{for $|v|\geq 1$}.  
	\end{aligned}
	\end{equation}
\end{lemma}
For the terms involving $N_1$, we can extract extra decay from the fact that any $f\in  \Omega^{A,\delta}_{\tilde{R},\delta_1,R,\eps}$ has zero average. 
\begin{lemma} \label{N1lemma}
	Let $N_1$ be given by \eqref{N1}. Assume that $h\in L^2_{\tilde{\lambda}}$  satisfies $|h(v)|\leq R_1 e^{-\frac12|v|}$ almost everywhere and has mean zero:
	\begin{align} \label{eq:meanzero}
		\int h(v) \ud{v} = 0.
	\end{align}
	For $a>0$, $\eps \in (0,\frac1{a}]$, $z=a+i\omega ,p= a+i\theta \in \Complex$ we have:
	\begin{equation} \label{N1est}
	\begin{aligned}
	\left| \int_{\Reals^3 \setminus B_1} \langle V, N_1(\eps,z,p,v')W \rangle h(v-v')  \ud{v'}\right| 		
	\leq  	\frac{C R_1 |V||W|(1+\eps |\omega-\theta|)}{(1+|v|^3) (1+\alpha(\eps z,v))^2(1+ \alpha(\eps p,v))^2}. 
	\end{aligned}	
	\end{equation}
	Similarly, for $L_1$ and $N_2$ we have the bounds:
	\begin{align} \label{eq:L1bd}
			\left| \int_{\Reals^3\setminus B_1} \langle V, L_1(\eps,z,p,v')W \rangle h(v-v')  \ud{v'}\right| 		
		 &\leq 	\frac{C R_1 |P_{v}^\perp V| |P_{v}^\perp W| (1+ \eps |	\theta-\omega|)}{(1+|v|)(1+\alpha(\eps z,v))(1+ \alpha(\eps p,v))}, \\
		\left| \int_{\Reals^3\setminus B_1} \langle V, N_2(\eps,z,p,v')W \rangle h(v-v')  \ud{v'}\right| 		
		&\leq  	\frac{C R_1|V||W|}{1+|v|^3} \frac{\eps^2 |p| |z| +\eps^2|p| |z| 				(1+\eps|\theta-\omega|)}						{(1+\alpha(\eps z,v))^2(1+ \alpha(\eps p,v))^2}. \label{eq:N2bd}		
	\end{align}
\end{lemma}
\begin{proof}
	The estimates \eqref{eq:L1bd} and \eqref{eq:N2bd} follow directly from
	\eqref{est:L1bound} and \eqref{est:N2bound} respectively.

	In order to prove \eqref{N1est}, we use the mean zero property \eqref{eq:meanzero} to obtain:
	\begin{align} \label{N1convest}
	\int_{\Reals^3 \setminus B_1} N_1(\eps,z,p,v')  h(v-v')  \ud{v'} 	
	= 	\int_{\Reals^3 \setminus B_1} \left(N_1(\eps,z,p,v')-N_1(\eps,z,p,v)\right)  h(v-v')  \ud{v'}. 	
	\end{align}
	For $|v|\leq 1$ the estimate is straightforward. For $|v|,|v'|\geq 1$, we
	estimate the difference above using the mean value theorem:
	\begin{align*}
		|N_1(\eps,z,p,v')-N_1(\eps,z,p,v)|\leq C \max_{v_1,v_2\in \{v,v'\}} \frac{|v-v'|(1+\eps|\theta-\omega|)}{(1+|v_1|^3)(1+\alpha(\eps z,v_2))(1+\alpha(\eps p,v_2))}.
	\end{align*}
	We plug this estimate and the assumption $|h(v)|\leq R_1 e^{-\frac12|v|}$
	into \eqref{N1convest} to conclude the proof.
\end{proof}
\begin{lemma}
	The following integral bound holds:
	\begin{equation} \label{loggain}
	\begin{aligned}
	\int_{\Reals} \frac{R\eps |\tau|}{(1+\eps|\tau|)(1+|\tau|)^2}\ud{\tau} &\leq C R \eps^\frac12.
	\end{aligned}
	\end{equation}
\end{lemma}
\begin{lemma}\label{lem:Lambda0}
	Recall $\Lambda_0$ introduced in \eqref{Lambda0def}. For $f\in  \Omega^{A,\delta}_{\tilde{R},\delta_1,R,\eps}$ we have:
	\begin{align*}
		|\Lambda_0[f](z,\tau,v)| \leq \frac{C(A)}{1+|\tau|^2}\left(\delta+ \frac{R\eps |\tau|}{1+|\eps \tau|}\right) \frac{e^{-\frac12 |v|}}{1+|z|}.
	\end{align*}
\end{lemma}
\begin{proof}
	Follows from Lemma~\ref{lem:realFT} and the explicit form of $M_1$ and $M_2$ (cf. \eqref{defM}).
\end{proof}

Now we are in the position to finish the proof of Lemma~\ref{lem:nonlinhigh}.
We use the representation of $Q^{\alpha,\alpha}_{\eps,A}[f](u)$ in \eqref{eq:Qfrep} and split it into three parts:
\begin{equation} \label{eq:Qaasplit}
	\begin{aligned}
	(2\pi)^\frac12 Q^{\alpha,\alpha}_{\eps,A}[f](u)	
	=	&\int_{\Reals} \int_\Reals \int \langle \lambda   \La( \nabla D^\alpha u) (z), \Lambda[f](\eps z, 				\omega-\theta) \La(\nabla D^\alpha u)(p)\rangle \ud{v}\ud{\theta} \ud{\omega} \\
	+	&	\int_{\Reals} \int_\Reals \int \langle  \La(D^\alpha u) (z) \nabla (\lambda), \Lambda[f](\eps z, 			\omega-\theta) \La(\nabla D^\alpha u)(p)  \rangle \ud{v}\ud{\theta} 	\ud{\omega} \\
	-	&\int_{\Reals} \int_\Reals \int \langle  \nabla(\La( D^\alpha u) \lambda)(z), \nabla \cdot  \Lambda[f](\eps 			z, \omega-\theta) \La(D^\alpha u)(p) \rangle \ud{v}\ud{\theta} 	\ud{\omega}  \\
	=	& J_1 + I_3 +J_2 .
	\end{aligned}
\end{equation}	
	The proof is very similar to the proof of Lemma~4.9 in \cite{velazquez_non-markovian_2018}.
	We estimate the terms $J_1$, $I_3$, $J_2$ separately, starting with $J_1$ To this end, 
	we write $V= \nabla D^\alpha \La(u)$ and recall $\Lambda_0$ (cf. \eqref{Lambda0def}) and  $\Lambda_1$, $\Lambda_2$ introduced in \eqref{Lambdas}.  We symmetrize in $z,p$, making use of the
	fact that the integrals are real valued:
	\begin{align*}	
		J_1 &= \int \langle \lambda(v) V  ,\Lambda[f](\eps z, \omega-\theta)  V\rangle  \ud{v} \ud{\omega} \ud{\theta} \\
		& = \int \langle \lambda(v) V  ,\Lambda_0[f](\eps z, \omega-\theta,v)  V \rangle  \ud{v} \ud{\omega} \ud{\theta} + \int \langle \lambda(v) V  ,\Lambda_s[f](\eps z, \omega-\theta)  V\rangle  \ud{v} \ud{\omega} \ud{\theta}  .
	\end{align*}
	The first term can be bounded by the dissipation functional $D^\alpha_{\eps,A}$ (cf. \eqref{dissipexplic}) using Lemma~\ref{lem:Lambda0}. Choosing $\delta(A)>0$ and $\eps(\delta,R,A,c)>0$ small enough, \eqref{loggain} yields:
	\begin{align}
		 \left|\int \langle \lambda(v) V  ,\Lambda_0[f](\eps z, \omega-\theta,v)  V \rangle  \ud{v} \ud{\omega} \ud{\theta}\right|\leq \frac{c}{6} D^\alpha_{\eps,A} .
	\end{align}
	We write the second term in terms of  $\Lambda_1$, $\Lambda_2$ introduced in \eqref{Lambdas}:
	\begin{align*}
		&\int \langle \lambda V  ,\Lambda_s[f](\eps z, \omega-\theta)  V\rangle  \ud{v} \ud{\omega} \ud{\theta}  	\\
		= &\int \langle \lambda V , (\Lambda_1+\Lambda_2)[f](\eps z, \omega-\theta)  V\rangle  \ud{v} \ud{\omega} = I_1+ I_2.
	\end{align*}
	We use \eqref{eq:L1bd} and  Lemma~\ref{lem:realFT}  to estimate:	
	\begin{align*}
	|I_1| \leq 	&C \int_{\Reals} \int_\Reals \int \int \lambda 	\frac{|P_{(v-v')}^\perp V(z,v)||P_{(v-v')}^\perp V(p,v)|}{|v-v'|} \frac{(1+ \eps |\theta-\omega|)|\La(f)(i(\theta-\omega),v')|}{(1+\alpha(\eps z,v-v'))(1+ \alpha(\eps p,v-v'))} \notag \\
	\leq 	&C(A) 	\int_{\Reals} \int_\Reals \int  	\frac{\lambda |V(z,v)|_v |V(p,v)|_v}{(1+|v|)(1+\alpha(\eps z,v))(1+ \alpha(\eps p,v))}  \frac{R \eps |\theta-\omega|} {(1+ \eps|\theta-\omega|) (1+|\theta-\omega|)^2}  \\
	+ 		&C(A)  \int_{\Reals} \int_\Reals \int  	\frac{\lambda |V(z)|_v |V(p)|_v}{(1+|v|)(1+\alpha(\eps p,v))(1+ \alpha(\eps z,v))}  Y_{R,\eps,\delta}(\theta-\omega).
	\end{align*}
	Choosing $\delta(A)>0$ and $\eps(\delta,R,A,c)>0$ small enough, \eqref{loggain} yields:
	\begin{align} \label{est:I1}
		|I_1|\leq \frac{c}{6} D^\alpha_{\eps,A}.
	\end{align}
	The estimate for $I_2$ follows similarly: After splitting into
	\begin{align*}
			|I_2| 	\leq			&\left|\int_{\Reals} \int_\Reals \int \langle \lambda V(z,v) (z,v) N_1(\eps,z,p,v) V(p,v)
		\rangle  		\La(f)(i(\theta-\omega),v') 	\ud{v'}	\ud{v}\ud{\theta} 									\ud{\omega}	\right| 		 \\
		+		&\left|\int_{\Reals} \int_\Reals \int \langle \lambda V (z,v) (z,v), N_2(\eps,z,p,v) V(p,v) 		\rangle \La(f)(i(\theta-\omega),v') 	\ud{v'} \ud{v}\ud{\theta} 									\ud{\omega}	\right| 		= I_{2,1} + I_{2,2}, 
	\end{align*}
	we can estimate $I_{2,1}$ using Lemma~\ref{N1lemma} and \eqref{loggain}, and $I_{2,2}$ using
	\eqref{eq:N2bd} and \eqref{loggain}. Therefore, we obtain $|I_2|\leq \frac{c}{6} D^\alpha_{\eps,A}$, and in combination with \eqref{est:I1} this yields:
	\begin{align} \label{eq:J1}
		|J_1| \leq \frac{c}{2} D^\alpha_{\eps,A}.
	\end{align}
	Next we estimate $I_3$. After an integrating by parts (we use the shorthand $W(z,v)=\La( D^\alpha u)(z,v)$) the term reads:
	\begin{align*} 
	I_3 =	&	-\int_{\Reals} \int_\Reals \int \langle W(z,v) \nabla^2 (\lambda), \Lambda(\eps z,\omega-\theta,v) W(p,v)  \rangle \ud{v}\ud{\theta} 	\ud{\omega}\\
	&-\int_{\Reals} \int_\Reals \int \langle  W(z,v)  \nabla (\lambda), \nabla \cdot\Lambda(\eps z,\omega-\theta,v) W(p,v)  \rangle \ud{v}\ud{\theta} 							\ud{\omega} \\				
	&	-\int_{\Reals} \int_\Reals \int \langle \nabla W(z,v) \otimes \nabla (\lambda), \Lambda(\eps z,\omega-\theta,v) W(p,v)  \rangle \ud{v}\ud{\theta} 	\ud{\omega} . 
	\end{align*}
	The first two lines are bounded by $\frac14 \|u\|^2_{V^n_{A,\lambda}}$ 
	by Lemma~\ref{lem:higherterms}. For the estimate of the last line we use $\lambda(v)=P_v \lambda(v)$. Then we recall
	the definition of $\Lambda$ (cf. \eqref{Lambdadef}) and apply \eqref{matrixest2} and Lemma~\ref{lem:realFT} to get:
	\begin{align}
		\left| \nabla \lambda \Lambda[f](\eps z,\tau,v)\right| \leq C(A) \frac{(1+\alpha(z,v)) B_2(z,v)(V,P_v (\nabla \lambda))}{(1+|\tau|^2)} (\delta+\frac{R\eps |\tau|}{(1+\eps|\tau|)}).
	\end{align}
	With Young's inequality and choosing $\eps,\delta>0$ small enough, we estimate $I_3$ by:
	\begin{equation} \label{eq:I3}
	\begin{aligned}
		|I_3|	&\leq \frac12 \|u\|^2_{V^n_{A,\lambda}} + \left|\int_{\Reals} \int_\Reals \int \langle \nabla W(z,v) \otimes \nabla (\lambda), \Lambda(\eps z,\omega-\theta,v) W(p,v)  \rangle \ud{v}\ud{\theta} 	\ud{\omega} \right| \\
				&\leq \frac{c}{4}D^\alpha_{\eps,A} + \frac12 \|u\|^2_{V^n_{A,\lambda}}. 
	\end{aligned}
	\end{equation}
	To bound $J_2$ we  apply Lemma~\ref{lem:higherterms} to obtain:
	\begin{align}\label{eq:J2}
		|J_2| &\leq \frac{c}{4}D^\alpha_{\eps,A} + \frac12 \|u\|^2_{V^n_{A,\lambda}}. 
	\end{align}
	Finally, we return to the decomposition \eqref{eq:Qaasplit}. From the estimates
	\eqref{eq:J2}, \eqref{eq:I3} and \eqref{eq:J1} we get the estimate:
	\begin{align*}
		(2\pi)^\frac12 Q^{\alpha,\alpha}_{\eps,A}[f](u) \leq cD^\alpha_{\eps,A} +  \|u\|^2_{V^n_{A,\lambda}},
	\end{align*}
	as claimed in the statement of the Lemma.
\end{proofof}

\subsection{Proof of Lemma~\ref{thm:invariance}}
\begin{proofof}[Proof of Lemma~\ref{thm:invariance}]
The proof presented here follows along the same lines as the proof of Theorem~4.10 in \cite{velazquez_non-markovian_2018}.

\noindent \textbf{Step 1: Picking $A,\delta>0$}	

We start by picking $A,\delta,\delta_0>0$ as in Theorem~\ref{thm:apriori}. 
Then for all $R,\tilde{R}>0$, $\delta_1\in(0,\delta_0)$, for  $\eps\in (0,\eps_0)$ small enough and $\gamma,\delta_2 \in(0,\frac12)$ arbitrary the mapping $\psi_{\delta_1}$ satisfies
 \begin{align*} 
 \|\Psi_{\delta_1} (f)\|_{V^{n}_{A,\lambda}} &\leq 1, \quad 	\|\partial_t \Psi_{\delta_1} (f)\|_{V^{n-2}_{A,\lambda}} \leq 1.
 \end{align*}
 By definition of $\kappa_{\delta_1}$, the mapping $\psi_{\delta_1}$ maps into $\Gamma^{n}_{A,\delta_1}$. We will pick the remaining coefficients in the
 order $\tilde{R},\delta_1, R$, and finally $\delta_2,\eps$ small enough to
 show that $\Omega^{A,\delta}_{\tilde{R},\delta_1,R,\eps}$ is left invariant by $\psi_{\delta_1}$.
The invariance is proved using equation \eqref{eq:fixedmollified} in Laplace variables (dropping the smoothing parameter $\gamma>0$ for brevity):
\begin{equation} \label{eq:bootstrap}
\begin{aligned}	
 \La(\partial_t u_\eps ) =z &\La ( u_\eps-u_0 ) = 	  \nabla \cdot \left( \int (M_1+M_2)(\eps z,v-v') u_0(v')  \nabla \La(u_\eps \kappa_1)(z,v) \ud{v'} \right) \\
-&  \nabla \cdot \left( \int  (M_1+M_2)(\eps z,v-v')\nabla u_0(v')   \La( u_\eps \kappa_1)(z,v) \ud{v'} \right) \\
+&  \nabla \cdot \left( \int (M_1+M_2)(\eps z,v-v') u_0(v')   \La\left(f(\cdot,v')  \nabla u_\eps(\cdot,v) \kappa_1\right)(z) \ud{v'}\right)  \\
-&  \nabla \cdot \left( \int  (M_1+M_2)(\eps z,v-v')\La\left(\nabla f(\cdot,v') u_\eps(\cdot,v)\kappa_1\right)(z) \ud{v'}  \right).
\end{aligned}
\end{equation}
Note that we can localize $u_\eps \kappa_1$ in time using the Volterra structure of the equation. 
We will make use of the equation above to bootstrap the apriori estimate on $\La(u_\eps)$ to pointwise estimates and thus show invariance of the set $\Omega^{A,\delta}_{\tilde{R},\delta_1,R,\eps}$ under $\psi_{\delta_1}$. 
	For our choice of $A,\delta>0$ we have:
	\begin{align*}
		\|u_\eps\|_{V^n_{A,\lambda}}+\|\partial_t u_\eps\|_{V^{n-2}_{A,\lambda}} \leq C.
	\end{align*}
	Using the compact support of $f \in \Omega^{A,\delta}_{\tilde{R},\delta_1,R,\eps}$ and
	$u_\eps \kappa_1$ we conclude, possibly changing $C(A,\delta)$:
	\begin{align*}
	\|f\|_{V^n_{0,\lambda}} +\|\partial_t f\|_{V^{n-2}_{0,\lambda}} \leq C(A,\delta),\quad  \|u_\eps \kappa_1\|_{V^n_{0,\lambda}} +\|\partial_t (u_\eps \kappa_1)\|_{V^{n-2}_{0,\lambda}} \leq C(A,\delta).
	\end{align*} 
	In particular we can bound:
	\begin{align*}
	\|\|\La(f)(z,\cdot)\|_{H^{n-2}}\|_{L^1_{\Re(z)=0}}\leq C(A,\delta), \quad  \|\|\La((u_\eps-u_0)\kappa_1)(z,\cdot)\|_{H^{n-2}}\|_{L^1_{\Re(z)=0}}\leq C(A,\delta).
	\end{align*}
	Using the Laplace representation \eqref{eq:bootstrap} we conclude:
	\begin{align} \label{eq:L1Linfty}
	|\|\La((u_\eps-u_0)\kappa_1)(z,\cdot)\|_{H^{n-2}}\|_{L^1_{\Re(z)=0}}+\|\|\La((u_\eps-u_0)\kappa_1)(z,\cdot)\|_{H^{n-2}}\|_{L^\infty_{\Re(z)=0}}\leq C(A,\delta).
	\end{align}

\noindent \textbf{Step 2: Picking $\tilde{R}>0$}	

	With the notation introduced in \eqref{eq:convolution} we have
	\begin{align*}
		\La\left(f(\cdot,v')  u_\eps(\cdot,v) \kappa_1\right)(z) = \left(\La(f)(\cdot,v') *_a \La (u_\eps\kappa_1)(\cdot,v)\right)(z).
	\end{align*}
	Now plugging \eqref{eq:L1Linfty} into \eqref{eq:bootstrap} and applying Young's inequality yields:
	\begin{align*}
		 \sup_{t\in[0,1]}\|\partial_t  (u_\eps-u_0) \|_{H^{n-4}_{\tilde{\lambda}}}\leq C(A,\delta).
	\end{align*}
	Now since $u_\eps(0,\cdot)-u_0(\cdot) =0$, we can pick $\tilde{R}$ such that:
	\begin{align} \label{eq:Rtilde}
		 \sup_{t\in[0,1]}\|\partial_t  \psi_{\delta_1}(f)(t,\cdot)\|_{H^{n-4}_{\tilde{\lambda}}}\leq \tilde{R}.
	\end{align}

\noindent \textbf{Step 3: Picking $\delta_1,R>0$}	

We now define the functions $p_\eps$, $q_\eps$ given by:
\begin{equation} \label{eq:decomposition}
\begin{aligned}
\partial_t p_\eps = &\nabla \cdot \left( \int_0^t \int \frac{e^{-s/\eps}P_{v'}^\perp}{\eps} (u_0+f)(t-s,v-v') \nabla u_\eps(t-s,v) \ud{v'} \ud{s}		 \right) \\
-&\nabla \cdot \left( \int_0^t \int\frac{e^{-s/\eps}P_{v'}^\perp}{\eps} \nabla(u_0+f)(t-s,v-v')u_\eps(t-s,v)  \ud{v'}\ud{s}		 \right),\quad p_\eps(0)=u_0, \\
z\La(q_\eps) =&\nabla \cdot \left(  \int M_2(\eps z,v')  \La(u_0+f)(\cdot,v-v')\nabla u_\eps(\cdot,v))(z) \ud{v'} 		 \right) \\
-&\nabla \cdot \left(  \int M_2(\eps z,v') \nabla \La(u_0+f)(\cdot,v-v')u_\eps(\cdot,v))(z)  \ud{v'}	 \right).
\end{aligned}  
\end{equation}
Making use of the explicit Laplace transform:
\begin{align*}
\frac{\pi^2}{4 \eps}\La(e^{-t/\eps}) P_v^\perp = M_1(\eps z,v), 
\end{align*}
we infer the identity 
\begin{align}
	u_\eps = p_\eps + q_\eps.
\end{align} 
Further, we observe that plugging the estimate \eqref{eq:L1Linfty} into \eqref{eq:bootstrap}
yields: 
\begin{align} \label{eq:R1}
|\nabla^m\La((u_\eps-u_0)\kappa_{\delta_1})|&\leq \frac{C(A,\delta_1)e^{-\frac12|v|}}{|1+\eps z|(1+|z|^2)}, \quad m\leq n-6 \\
|\nabla^m\La(q_\eps \kappa_{\delta_1})|&\leq \frac{C(A,\delta_1)\eps|z|e^{-\frac12|v|} }{|1+\eps z|(1+|z|^2)}, \quad m\leq n-6  \label{eq:q}	.
\end{align}
In particular we can estimate:
\begin{align}
	|\partial_t \nabla^ m q_\eps| &\leq C e^{-\frac12 |v|} .
\end{align}
To estimate $p_\eps$, we decompose it further. To this end, let $b$ be the function given by:
	\begin{align*}
		b(t,r):= \frac{e^{-tr}}{r^2}  + \frac{t}{r} - \frac{1}{r^2},
	\end{align*}
	which satisfies the equation
	\begin{equation}
		\partial_t b(t,r) = \frac{1-e^{-rt}}{r}, \quad \partial_{tt} b(t,r) = e^{-rt}, \quad b(0,v) = 0 .
	\end{equation}
	Now for any $u_0 \in H^n_{\lambda}$, we define a boundary layer function $B(t,v;u_0)= \nabla \cdot B_F(t,v;u_0)$ by:
	\begin{equation} \label{defBL}
	\begin{aligned}
	B_F(t,v;u_0) :=\int  \frac{\pi^2}{4}\frac{b(t,\frac{|v'|}{\eps})P_{v'}^\perp}{\eps} 
	\left( u_0(v-v') \nabla u_0(v)  -  \nabla u_0(v-v')  \ u_0(v) \right) \ud{v'}.
	\end{aligned}
	\end{equation}
	By construction, $B=\nabla \cdot B_F$ then satisfies:
	\begin{align*}
		\partial_{tt} B(t,v) &=   \nabla \cdot \left(\int \frac{\pi^2}{4} \frac{e^{-\frac{t|v'|}{\eps}}P_{v'}^\perp}{\eps}  \left( u_0(v-v')  \nabla 									u_0(v)
		-    \nabla u_0(v-v')  u_0(v)\right) \ud{v'} \right), \\
		B(0,v)	&=0  \quad  \partial_t B(0,v)=0.
	\end{align*}
	If $u_0=m(v)$ is a Maxwellian distribution, the boundary layer $B$ vanishes.
	To see this, we note that $\nabla m(v) = - \frac{v}{\sigma^2} m(v)$, hence
	\begin{align*}
		P_{v'}^\perp \left(  m(v-v')\nabla m(v) - \nabla m(v-v')   m(v)\right) 
		=-	P_{v'}^\perp  \frac{v'}{\sigma^2}  m(v-v')m(v) = 0,  
	\end{align*}
	and 
	\begin{equation} \label{eq:maxboundary}
	B(t,v;m) = 0.
	\end{equation}	
	Computing the second time derivative of $p_\eps$, we find:
	\begin{align*}
	\partial_{tt} p_\eps 	= 	& \nabla \cdot \left(\int_0^t \int\frac{\pi^2}{4}  \frac{e^{-\frac{s|v'|}{\eps}}P_{v'}^\perp}{\eps } \partial_t((u_0+f)(t-s,v-v')  \nabla 									u_\eps(t-s,v)) \ud{v'} \ud{s} \right) \\
	-  	& \nabla \cdot \left(\int_0^t \int \frac{\pi^2}{4} \frac{e^{-\frac{s|v'|}{\eps}}P_{v'}^\perp}{\eps } \partial_t(\nabla (u_0+f)(t-s,v-v'))  u_\eps(t-s,v) \ud{v'} \ud{s} \right) \\
	+	& \nabla \cdot \left(\int \frac{\pi^2}{4} \frac{e^{-\frac{t|v'|}{\eps}}P_{v'}^\perp}{\eps}  \left( u_0(v-v') \nabla 						u_0(v) - \nabla u_0(v-v')    u_0(v)\right) \ud{v'} \right) \\
	=  &R_1 + R_2 + \partial_{tt} B .
	\end{align*}
	Since $\sup_{t\in[0,1]}    \|\partial_t  f(t,\cdot)\|_{H^{n-4}_{\tilde{\lambda}}}\leq \tilde{R}$ by assumption, we obtain for $m\leq n-6$:
	\begin{align*}
	|\partial_{tt} \nabla^m(p_\eps - B)(t,v)|=|\nabla^m(R_1(t,v) +R_2(t,v))| \leq C \tilde{R} e^{-\frac12 |v|},\quad \text{for $t\in[0,1]$},
	\end{align*}
	Combined with the lemma above this shows:
	\begin{align*}
	|\partial_{tt} (\nabla^m(p_\eps - B) \kappa_{\delta_1})| \leq C \tilde{R} e^{-\frac12 |v|} (1 + \frac{t}{\delta_1} + \frac{t^2}{\delta_1^2}) \kappa_{\delta_1},\quad \text{for $t\in[0,1]$}. 
	\end{align*}
	After integrating by parts twice this allows to bound the Laplace transform by:
	\begin{align*}
	|\La(\nabla^m(p_\eps - B) \kappa_{\delta_1})(z,v)| 	\leq \frac{C \tilde{R} e^{-\frac12 |v|}}{|z|^2} \int_0^\infty C e^{-\frac12 |v|} (1 + \frac{t}{\delta_1} + \frac{t^2}{\delta_1^2}) \kappa_{\delta_1} \ud{t}	\leq \frac{C \tilde{R} e^{-\frac12 |v|}}{|z|^2} \delta_1.
	\end{align*}
	This on the other hand implies that on the line $\Re(z)=a$ we have:
	\begin{align*}
	|\La(\nabla^m (p_\eps - B) \kappa_{\delta_1})(z,v)| &\leq \frac{\delta_1 C \tilde{R} e^{-\frac12 |v|}}{1+|z|^2}, \quad \text{for }m\leq n-6.
	\end{align*}
	Now we first pick $\delta_1>0$ such that: 
	\begin{align} \label{eq:p1}
		|\La(\nabla^m (p_\eps - B) \kappa_{\delta_1})(z,v)| &\leq \frac12 \frac{\delta e^{-\frac12 |v|}}{1+|z|^2}, \quad \text{for }m\leq n-6,
	\end{align}
	and secondly $R>0$ depending on $\delta_1$ as the largest constant appearing on the right-hand sides of \eqref{eq:R1}-\eqref{eq:q}.  

	\noindent \textbf{Step 4: Picking $\delta_2>0$}	
	
	Finally, we consider the explicitly given function $B$. Since we consider initial data of the form $u_0=m+\delta_2 v_0$, we can make use of the fact that the Maxwellian does not contribute to the boundary layer (cf. \eqref{eq:maxboundary}) to obtain:
	\begin{align*} 
	|\La(\nabla^m B \kappa_{\delta_1})(z,v)| \leq C(\delta_1) \frac{\delta_2 e^{-\frac12 |v|}}{1+|z|^2}, \quad \text{for }m\leq n-6. 
	\end{align*}
	Now we choose $\delta_2>0$ small enough such that
	\begin{align} \label{eq:p2}
		|\La(\nabla^m B \kappa_{\delta_1})(z,v)|\leq \frac12  \frac{\delta e^{-\frac12 |v|}}{1+|z|^2}, \quad \text{for }m\leq n-6. 
	\end{align}
	Then we decompose $\psi_{\delta_1}(f)$ into the functions:
	\begin{align*}
		u_\eps 		&= u_{\eps,1} + u_{\eps,2} \\
		u_{\eps,1} 	&= p_\eps \kappa_{\delta_1} ,\quad u_{\eps,2} = q_\eps \kappa_{\delta_1} .
	\end{align*}
	Then for the choice of $A,\delta,\tilde{R},\delta_1,R,\delta_2>0$ the estimates
	\eqref{eq:p1}-\eqref{eq:p2} show 
	\begin{align*}
		\|u_{\eps,1}\|_{E_A} \leq \delta.
	\end{align*}
	Furthermore, we picked $R$ as the largest constant in \eqref{eq:R1}-\eqref{eq:q}, hence
	\begin{align}
		\|\psi_{\delta_1}(f)\|_{G_A}	&\leq R.\\
		\|u_{\eps,2}\|_{F_{\eps,A}} 		&\leq R.
	\end{align}	
	Since we chose $\tilde{R}>0$ to satisfy \eqref{eq:Rtilde} we conclude $\psi_{\delta_1}(f)\in \Omega^{A,\delta}_{\tilde{R},\delta_1,R,\eps}$.
\end{proofof}


\textbf{Acknowledgment.}
The author acknowledges support through the CRC 1060
\textit{The mathematics of emergent effects}
at the University of Bonn that is funded through the German Science
Foundation (DFG).

\bibliography{BBGKY_Landau_2019}
\bibliographystyle{plain} 

%
%



\end{document}